\documentclass[final]{IEEEtran}
\usepackage{amsthm,amssymb,graphicx,multirow,amsmath,color,amsfonts}
\usepackage[update,prepend]{epstopdf}
\usepackage[noadjust]{cite}
\usepackage{tikz}
\usepackage{setspace}
\usepackage{bbm} 
\usepackage{pdfpages}
\usepackage{multirow}
\usepackage{comment}
\usepackage[justification=centering]{caption}

\usepackage{lipsum, color}
\usepackage{tikz}
\usepackage{tkz-tab}
\usetikzlibrary{backgrounds,shapes.multipart}
\usetikzlibrary{shapes.geometric, arrows}
\usetikzlibrary{automata,arrows,positioning,calc}
\usepackage{subcaption} 
\usepackage{multicol}
\usepackage{booktabs,chemformula,lipsum}
\usepackage{authblk}
\usepackage{dblfloatfix}
\usepackage{pgfplots}
\usetikzlibrary{shapes.misc}
\pgfplotsset{compat=1.16}
\usepackage{breqn}
\newcommand{\colb}{\textcolor{blue}}

\usepackage{hyperref}
\usepackage{algorithm}
\usepackage[noend]{algpseudocode}
\usepackage{algpseudocode}
\usepackage[most]{tcolorbox}

\usepackage{pifont}
\usepackage[utf8]{inputenc}
\usepackage[noadjust]{cite}
\usepackage{amsmath}
\usepackage{times}
\usepackage{textcomp}
\usepackage{blkarray}
\usepackage{amsfonts}
\usepackage{amssymb}
\usepackage{amsthm}
\usepackage{graphicx}
\usepackage{float}
\usepackage{layout}
\usepackage{array}
\usepackage{comment}
\usepackage{flushend}
\usepackage{color}
\usepackage{xcolor}
\usepackage{soul}
\usepackage{footmisc}
\usepackage{arydshln}
\theoremstyle{plain}
\usepackage{epstopdf}
\usepackage{amssymb}
\usepackage{bbm}
\usepackage{bm}
\usepackage{mathtools}
\usepackage{cuted}
\usepackage{lipsum, color}
\usepackage{tikz}
\usepackage{tkz-tab}
\usetikzlibrary{backgrounds,shapes.multipart}
\usetikzlibrary{shapes.geometric, arrows}
\usetikzlibrary{automata,arrows,positioning,calc}
\usepackage{caption}
\usepackage{subcaption} 
\usepackage[inline]{enumitem}
\usepackage{multicol}
\usepackage{booktabs,chemformula,lipsum}
\usepackage{authblk}
\usepackage{pgfplotstable}
\usepackage{scalerel}
\usepackage{lipsum}
\usetikzlibrary{plotmarks}
\tikzstyle{vertex}=[circle,fill=green!25,minimum size=20pt,inner sep=0pt]
\tikzstyle{selected vertex} = [vertex, fill=red!24]
\tikzstyle{edge} = [draw,thick,-]
\tikzstyle{weight} = [font=\small]
\tikzstyle{selected edge} = [draw,line width=5pt,-,red!50]
\tikzstyle{ignored edge} = [draw,line width=5pt,-,black!20]
\tikzstyle{startstop} = [rectangle, rounded corners, minimum width=3cm, minimum height=1cm,text centered, draw=black, fill=red!30]
\tikzstyle{io} = [trapezium, trapezium left angle=70, trapezium right angle=110, minimum width=3cm, minimum height=1cm, text centered, draw=black, fill=blue!30]
\tikzstyle{process} = [rectangle, minimum width=3cm, minimum height=1cm, text centered, draw=black, fill=orange!30]
\tikzstyle{decision} = [diamond, minimum width=3cm, minimum height=1cm, text centered, draw=black, fill=green!30]
\tikzstyle{arrow} = [thick,->,>=stealth]
\usetikzlibrary{decorations}
\usetikzlibrary{positioning,arrows.meta,,bending}

\def\nb0{{\mathbf{0}}}
\def\nb1{{\mathbf{1}}}

\newtheorem{lemma}{Lemma}

\newtheorem{definition}{Definition}

\newtheorem{theorem}{Theorem}

\allowdisplaybreaks 
\begin{document}

\graphicspath{{./Figures/}}
\title{Spatial Distribution of Mean PAoI in LEO Satellite-enabled IoT Networks}
\author{
Badiaa Gabr and Mustafa A. Kishk
\thanks{Badiaa Gabr and Mustafa A. Kishk are with the Department of Electronic Engineering, Maynooth University, W23 F2H6, Ireland. Email: \{badiaa.gabr.2025@mumail.ie,\ mustafa.kishk@mu.ie\}. 
} 
}

\maketitle
\begin{abstract}
Non-terrestrial networks (NTNs), such as Low Earth Orbit (LEO) satellites, can provide consistent coverage for Internet of Things (IoT) nodes, particularly in rural areas where cellular coverage is lacking. These IoT nodes are used to monitor and record their surroundings, supporting various applications, including real-time applications. In such cases, transmission of the most up-to-date data is essential, making the Age of Information (AoI) a critical metric to consider. The AoI measures the elapsed time from the data generation at the source node to its reception at the destination. Moreover, the Peak AoI (PAoI), which denotes the maximum AoI in the system before an update is received, provides valuable insights into the network's performance. Therefore, this encourages us to focus on the PAoI within a large-scale and fully loaded IoT network, which is covered by LEOSATs. Using tools from stochastic geometry, this paper derives the spatial distribution of the uplink mean PAoI in a LEOSAT-enabled IoT network. To achieve that purpose, the meta distribution of the uplink conditional coverage probability is derived. Numerical results articulate how varying the number of deployed LEOSATs and their altitudes affects the coverage probability meta distribution and the spatial distribution of the mean PAoI. Additionally, by examining various scenarios with different numbers of LEOSATs and altitudes, we can derive valuable insights to determine the optimal beamwidth for LEOSATs tailored to each specific scenario to minimize the conditional mean PAoI.


\end{abstract}


\vspace{-2mm}
\section{Introduction}\label{Sec:Introduction}

Alongside traditional performance metrics such as transmission latency, data rate, and packet loss, data freshness emerges as a vital metric for assessing network reliability, particularly in real-time applications. Data freshness refers to how current and relevant the transmitted data is. Accordingly, it is challenging to account for in large-scale networks such as the Internet of Things (IoT) networks, where a massive amount of data is exchanged periodically within the network. To assess data freshness, the concept of Age of Information (AoI) has been introduced, which measures the time from data generation to its reception \cite{AoI1}. In addition to the Age of Information (AoI), the Peak Age of Information (PAoI) is a crucial metric to be considered, representing the maximum AoI in the network before an update is received \cite{AoI8}. Considering two distinct transmission approaches: the first-come-first-served (FCFS) and the last-come-first-served with preemption (LCFS-PR) is powerful to investigate both peak and average AoI within a cellular network \cite{AoI23}. Therefore, it is vital to determine the PAoI of data collected from IoT nodes to assess network reliability. Moreover, the IoT nodes are sometimes distributed in rural regions, in which cellular network infrastructure is commonly limited, and hence cannot support the IoT nodes in this case. To ensure a consistent connection for IoT nodes in rural and remote areas, some of the terrestrial transmissions are being shifted to non-terrestrial networks (NTN), such as Low Earth Orbit (LEO) satellites. LEOSATs provide a larger coverage area than terrestrial networks, enabling more effective transmission of data collected by IoT nodes. 
On the other hand, the rotational motion of LEOSATs increases the randomness of the distances to ground IoT nodes. This randomness increases the computational complexity of network performance metrics, specifically when the number of transmissions is high due to the massive number of connected devices. Therefore, this scenario highlights the need for a tractable model. One effective mathematical framework for achieving this is stochastic geometry \cite{Haenggi2012,Jeffrey}, which enables the evaluation of distinct performance metrics for terrestrial and non-terrestrial networks, such as coverage probability, data rate, and AoI of the transmitted data. Multiple stochastic geometry-based models for modeling LEOSAT networks are proposed, considering distinct random processes like the Poisson point process (PPP) \cite{PPP,Niloofar2,Badiaa_LEO1,Akram2}, binomial point process (BPP) \cite{Anna1,Satellite4,Ruibo3,Satellite5,Satellite7}, and line Cox Process \cite{Choi1,Choi2,Cox1,Cox2}. Furthermore, these random processes play a crucial role in determining the distribution expressions of all the performance metrics mentioned earlier, ensuring the system's tractability. 
Accordingly, in this work, we focus on real-time IoT applications that impact various aspects of daily life, including wearable devices, traffic monitoring, crop health, and weather conditions. Transmitting the most up-to-date data is essential for ensuring the reliability of these systems.  Therefore, the concept of data freshness is introduced to address the question: How old is this data? Sometimes, these IoT sensors are deployed in rural areas that may experience cellular coverage interruptions but are still covered by the LEOSAT network. Since the freshness of transmitted data is critical for these applications, we utilize the PAoI metric, which measures the maximum AoI in the system since the last data reception, representing the worst-case scenario for AoI. This metric indicates the worst-case scenario for AoI observed in the system. Consequently, this work employs tools from stochastic geometry to analyze the spatial distribution of the uplink conditional mean PAoI for a given distinct number of LEOSATs and their altitude.


\subsection{Related Work}
In this section, we discuss the most related works in the literature, which can be categorized into: \begin{enumerate*} [nosep]
    \item Stochastic geometry-based analysis in LEOSAT networks, 
    \item AoI analysis in LEOSAT networks,
    \item AoI analysis based on stochastic geometry in LEOSAT networks, and
    \item SIR/SINR meta-distribution in LEOSAT networks
\end{enumerate*}\\

\textbf{Stochastic geometry-based analysis in LEOSAT networks.} 
Stochastic geometry has been employed to model the spatial distribution of LEOSATs to derive multiple performance metric expressions, such as coverage probability, AoI, data throughput, and power consumption in non-terrestrial networks. Hence, the authors in \cite{Anna4} assumed the locations of the  LEOSATs follow BPP. They considered ground gateways that acted as relays between the LEOSATs and ground users to derive coverage probabilities for both the gateways and the LEOSATs, for different numbers of satellites and varying altitudes. In \cite{Satellite4}, the authors formulated the exact expression of the local delay of downlink retransmission from the LEOSATs to the ground users, considering that the LEOSATs were located following BPP. Authors in \cite{Akram3} optimized the uplink coverage probability by determining the optimal LEOSATs' altitude and their effective beamwidth, where both the LEOSATs and ground nodes followed two independent BPPs. A heterogeneous satellite network: closed access and open access modeled as a Cox process in \cite{Cox1} to determine the downlink coverage probability, where the performance of the open access LEO network outperformed the closed access one. Authors in \cite{Cox2} proposed a Cox-based framework to consider distinct performance metrics: average fraction of time that the LEOSAT covered the typical user, average amount of data transmitted from the typical user to its serving LEOSAT, and the delay distribution in the system. The LEOSAT network was located as a PPP to find the average downlink success probability considering distinct SINR thresholds and LEOSAT transmitted power in \cite{Akram1}. A tractable model was proposed in \cite{Namyoon1} to compare the downlink coverage probabilities of PPP, BPP, and actual Starlink constellations. To validate the accuracy of using stochastic processes to describe the locations of the LEOSATs, the authors in \cite{Akram4} assumed that the locations of the LEOSATs follow PPP. They compared the coverage probabilities of two practical orbital patterns: Walker-Delta and Walker-Star, with the coverage probability determined by the PPP. Their findings revealed that these probabilities are nearly identical.

\begin{table*}[htbp]\caption{Table of Notations}
\centering
\begin{center}
    \begin{tabular}{| l | p{13cm} |}
    \hline
    \textbf{Notation} & \textbf{Description} \\ \hline
    $\Psi_{s}$, $\gamma_{s}$ & PPP modeling the locations of the LEOSATs and their density. \\ \hline
    $\Psi_{n}$, $\gamma_n$ & PPP modeling the locations of ground IoT nodes and their density.  \\ \hline
    $r_s = r_e + a_s$ & Radius of the LEOSATs sphere, $r_e = 6371$ km is the Earth's radius, and $a_s$ is the LEOSATs' altitude.\\ \hline
    $r_{\rm max}$ & The maximum association distance from the ground node to the typical LEOSAT to establish a connection. \\ \hline
    $\phi^S$, $\theta^S$ & The zenith angle and the azimuth angle of the LEOSAT. \\ \hline
    $N_s$ & Number of LEOSATs on the surface of the LEOSAT sphere. \\ \hline
    $\theta_b$, $\theta_c$ & LEOSAT beamwidth angle and contact angle. \\ \hline
    $|{\mathcal A_{\rm Cap}}|$ & The spherical cap area with a radius $r$ and length $l$ that contains the typical LEOSAT above the horizon of nodes of interest.\\ \hline
    $\rho_t^{\rm IoT}$, $\sigma_S^2$ & The transmitted power of the IoT node, and the noise power at the typical LEOSAT. \\ \hline
    $r_g$ &  The Euclidean distance from ground node $g$ to the typical LEOSAT.\\ \hline
    $H_0$, $H_i$ & Shadowed Rician channel fading for the intended source node, and for the interferer node $i$.\\ \hline
    $c$, $f_c$, $\lambda_c$  & The speed of light, frequency spectrum, and wavelength of the transmitted signal.\\ \hline
    ${\rm SINR}$ & The signal to interference plus noise ratio.\\ \hline
    ${\mathcal I_{\rm agg}^{\rm U}}$, ${\mathcal I_{\rm agg}^{\rm D}}$ & The aggregate interference at the typical LEOSAT and the ground data controller. \\ \hline
    $P^{\rm U}_{\rm cov}(\gamma_{\rm th}|\gamma_n)$ & The conditional uplink coverage probability when the SINR at the typical LEOSAT exceeds a threshold $(\gamma_{\rm th})$.\\ \hline
    $P_{v}$  & Server's vacation probability.\\ \hline
    $\lambda_a$ & Arrival rate of the updates at the source node.\\ \hline
    $T^S_k$, $T^W_k$, $T^I_k$  & Service time of $k$th update, waiting time for the $k$th update to be served, and the interarrival time after the last successfully transmitted update till $k$th update is generated at the source node. \\ \hline
    $A_k$  & The Peak AoI (PAoI) of the transmitted $k$th update.\\ \hline
    \end{tabular}
\end{center}
\label{tab:TableOfNotations}
\end{table*}

\textbf{AoI analysis in LEOSAT networks}.
The AoI becomes a potential metric that should be considered within the NTNs, especially for real-time applications. The authors in \cite{AoI_LEO_01} proposed a reinforcement energy-aware algorithm to minimize the AoI for both uplink and downlink transmissions by applying a non-orthogonal multiple access (NOMA) scheme. A slotted Aloha access protocol was introduced in \cite{AoI_23} to derive a closed-form expression for the average AoI in scenarios involving both single and multiple ground node transmissions. In \cite{Satellite6}, the authors proposed a latest Information update algorithm (LIUA) approach that utilized a continuous-time-varying graph to identify the optimal pair of transmission delay and arrival rate. The authors in \cite{Tony_Quek} introduced a resource allocation algorithm that minimized the AoI and average transmission energy consumption, constrained by the user
pairing, transmitting power, and beamwidth. The authors in \cite{semantics} aimed to develop a unified framework that integrates both terrestrial and non-terrestrial networks. Their primary focus was on various metrics for assessing data freshness, including Age of Information (AoI), Version Age of Information (VAoI), Query Age of Information (QAoI), and Query Version Age of Information (QVAoI). These metrics are designed to evaluate the relevance and freshness of transmitted data.

\textbf{AoI analysis based on stochastic geometry in LEOSAT networks}. To guarantee the model tractability, stochastic geometry has been effectively employed to analyze AoI performance in LEOSAT networks by treating node locations as spatial point processes.
Authors in \cite{Badiaa_LEO1} considered the ground IoT nodes and LEOSATs that followed two independent PPPs to derive an expression for the mean PAoI of the received updates at the ground data controller from a distant ground IoT source node relying on the serving LEOSAT, considering various numbers of LEOSATs and their respective altitudes. The authors in \cite{Howard} presented an on-off process to model the service mechanism of the LEOSATs, enabling them to calculate the AoI of the transmitted updates. In \cite{AoI18}, the authors examined a hybrid distribution combining a PPP and a Poisson Cluster Process (PCP) to analyze how the distribution of ground nodes influences the AoI for status update packet transmissions facilitated by UAVs and satellites in a LEOSAT network. Moreover, the authors in \cite{Badiaa_LEO2} assumed the LEOSAT and ground IoT nodes followed two independent homogeneous PPPs (HPPPs) to examine the average uplink PAoI at the serving LEOSAT.  

\textbf{SIR/SINR Meta distribution in LEOSAT networks.} To examine the reliability of the LEOSAT network, the meta distribution of the SIR/SINR at the destination should be derived. SIR/SINR meta distribution determines the percentile of links that exceed a certain SIR/SINR threshold for a given network realization. In \cite{Meta_1}, the authors considered two system models: a simplified non-spherical (planar) model used in analysis and a spherical model used in the simulation. They assumed the ground users follow a 2D HPPP to determine the uplink SIR and SINR meta distributions and throughput distribution under Nakagami fading. A downlink SIR meta-distribution considering Nakagami fading for the transmitted data from LEOSATs to the ground users was derived in \cite{Meta_2}. The authors assumed that LEOSATs followed an HPPP to determine the moments of the downlink coverage probability for the users. A Shadowed Rician fading model was introduced in \cite{Meta_4} to derive the SINR meta distribution for a downlink transmission from the LEOSAT to the ground relay, where the LEOSATs followed a HPPP distribution. In \cite{Meta_3}, the authors focused on the SNR at the destination, neglecting the interference present at the destination to determine the end-to-end SNR meta distribution in a broadcasting direct-access LEOSAT-based IoT network. They assumed the LEOSATs follow a BPP distribution and that the channel power gain among the IoT nodes and the LEOSATs was characterized by shadowed Rician fading. Additionally, all of the previously mentioned works utilized a Beta distribution approximation to formulate the SIR/SINR meta distribution expression due to its accurate correspondence with the exact expression. The primary difference between this work and our previous work \cite{Badiaa_LEO1} lies in our objectives. In \cite{Badiaa_LEO1}, the main focus was on deriving the average PAoI. On the other hand, this work shifts towards a more detailed spatial analysis through the concept of meta-distribution, accounting for the reliability of the network constrained to different time deadlines. Therefore, we derive the b-th moments and the meta-distribution of the uplink coverage probability. As a result, we obtain the spatial distribution of the mean uplink PAoI. Consequently, deriving the PAoI distribution rather than its average offers a comprehensive understanding of the PAoI behaviour with respect to distinct system parameters such as the number of LEOSATs and the LEOSAT altitude. Additionally, it provides the time percentile required to achieve a specific reliability in the network. Furthermore, this framework can be extended to incorporate reconfigurable holographic surfaces (RHS) to investigate the effects of a hybrid beamforming design within multi-altitude LEO satellite networks \cite{El_Hajjar_1,El_Hajjar_2,El_Hajjar_3}.

\subsection{Contributions} \label{SubSec:Contributions}
This work introduces a novel model which, to our knowledge, has not been previously explored in the literature. The main contribution of this work is proposing a stochastic geometry-based spherical model to derive the spatial distribution of the uplink conditional mean PAoI of the transmitted update from the ground IoT node to its serving LEOSAT under shadowed Rician fading. Additionally, we derive the b-th moment of the conditional coverage probability and its meta distribution expression, considering the SINR at the serving LEOSAT to determine the b-th moment of the conditional uplink mean PAoI and its spatial distribution.

\textbf{PAoI in LEOSAT-enabled IoT}.
We consider a large-scale IoT network deployed in a rural area that lacks cellular coverage. To ensure consistent transmission, we assume the LEOSAT network covers this network. In this scenario, both the IoT nodes and the LEOSATs are modeled using two different and independent HPPPs: one representing the ground IoT nodes distributed on the Earth's surface and the other representing the LEOSATs deployed on the surface of a sphere, which is concentric with the Earth. This work considers an uplink transmission scenario, focusing on one of the essential metrics in real-time applications: the PAoI of the transmitted update at the serving LEOSAT. 

\textbf{Performance Analysis.}
Due to the ergodicity of the PPP, the spatial average is the ensemble mean. Therefore, for a given network realization, the b-th moment of the conditional uplink coverage probability expression is formulated across different network realizations. Consequently, we use this expression to formulate the b-th moment of the conditional uplink mean PAoI expression. Additionally, to account for the spatial variations of the coverage probability, which is a random variable from the spatial perspective domain, we find the Beta distribution approximated meta-distribution expression for the uplink conditional coverage probability by calculating the SINR at the serving LEOSAT. Based on this distribution, we find the spatial distribution of the conditional uplink mean PAoI. 

\textbf{System-level insights.} 
By analyzing the derived performance metrics and numerical results, we gain valuable system-level insights. We examine the impact of various parameters: the total number of LEOSATs and the LEOSAT altitude on the system performance metrics. We find that an optimal LEOSAT altitude can minimize the first moment of the conditional mean PAoI at the serving LEOSAT. 
Additionally, we analyze optimal values of the LEOSAT beamwidths to reduce the conditional mean PAoI across various deployment scenarios.

\subsection{Paper Outline}

The rest of the paper is organized as follows. Sec. \ref{Sec:System_Model} presents the proposed system model. The spatial distribution of the conditional mean PAoI is introduced in Sec. \ref{Sec:Conditional_PAoI_Analysis}. Sec. \ref{Sec:Numerical_Results} discusses simulation results and main insights. The conclusion of this work is in Sec. \ref{Sec:Conclusion}. A summary of all notations used in this paper is provided in Table \ref{tab:TableOfNotations}.

\vspace{-2mm}
\section{System Model} \label{Sec:System_Model} 
In this paper, we consider an IoT network deployed in a region lacking cellular infrastructure and hence relying solely on the LEOSAT network for connectivity. We focus on the uplink scenario where the IoT node transmits its updates to its associated LEOSAT (referred to as a typical LEOSAT). More details about the system model are provided in the following subsections.

\begin{figure*}[t]
\centering
\includegraphics[width=0.75\textwidth]{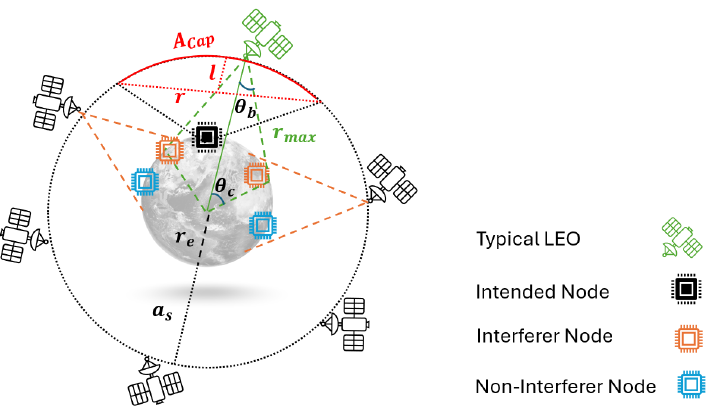}
    \caption{\small{The illustration of the source ground node association to its serving LEOSAT and locations of the rest of the ground nodes.}}
    \label{fig:System_Model}
\end{figure*}

\subsection{Spatial Model}
We consider a large-scale IoT network that is covered by the LEOSAT network, where the ground IoT nodes are distributed following a HPPP $\Psi_n$ in $\mathbb{R}^2$ with intensity $\gamma_n$. The LEOSATs are located on the surface of a sphere $\mathbb{S} \subset \mathbb{R}^3$ with radius $r_s = a_s + r_e$, which is concentric with the Earth, where $a_s$ and $r_e$ are the LEOSAT altitude and the Earth's radius, respectively. We consider the locations of LEOSATs follow another HPPP $\Psi_s$ with intensity $\gamma_s$, where $\Psi_n$ and $\Psi_s$ are independent. Therefore, $\gamma_s$ is determined by $\gamma_s=N_s/(4 \pi r_s^2)$, where $N_s$ refers to the total number of LEOSATs. Moreover, the location of any LEOSAT can be described by spherical coordinates $(r_s, \phi^s, \theta^s)$, where $\phi^s$ is the zenith angle and $\theta^s$ is the azimuth angle of the LEOSATs. Without loss of generality, the typical LEOSAT is located at $(r_s, 0, 0)$. Each LEOSAT has a coverage area, referred to as the footprint, which is determined by its beamwidth angle. Only the IoT nodes within this footprint can be served by the LEOSAT. Moreover, the IoT nodes located within the footprint of the typical LEOSAT and served by another LEOSAT are regarded as potential interferers, represented by the orange IoT nodes shown in Fig. \ref{fig:System_Model}. Let \(\theta_b\) represent the beamwidth angle of a typical LEOSAT. Consequently, the maximum distance at which ground nodes can establish a connection with the LEOSAT is denoted as ${ r_{\rm max}}$, and is calculated by ${r_{\rm max}}(\theta_c)=\sqrt{r_e^2+ r_s^2 -2 r_er_s\cos{(\theta_c)}}$, where $\theta_c$ is the contact angle of the LEOSAT and is defined as

\begin{align} 
\label{Eq:Contact_Angle}
\theta_c=
            \begin{cases}
                {\rm arcsin}\left(\dfrac{r_s}{r_e}\sin(\theta_b)\right), & \theta_b < {\rm arcsin}\left(\dfrac{r_e}{r_s}\right) \\
                {\rm arccos} \left(\dfrac{r_e}{r_s}\right), & \theta_b \ge {\rm arcsin}\left(\dfrac{r_e}{r_s}\right).
            \end{cases}
\end{align}

The IoT nodes are uniformly distributed within the footprint of the typical LEOSAT. In particular, any node with a zenith angle in $[0,\theta_c]$ will lie within this footprint and can establish a connection with the typical LEOSAT. Let $r_o$ be a random variable that refers to the contact distance between the intended node $o$ and the typical LEOSAT. Hence, its PDF is defined as

\begin{align}\label{Eq:Typical_Contact_Distance_PDF}
    f_{R_o}(r_o)=
        \begin{cases}
            \dfrac{r_o}{\theta_c r_e r_s} \dfrac{1}{\sqrt{1-\left(\dfrac{r_e^2+r_s^2-r_o^2}{2r_er_s}\right)^2}}, & \\
            \hspace{4cm} a_s < r_o \leq  {r_{\rm max}\rm}(\theta_c) \\
            0, &\hspace{-3cm} otherwise,
        \end{cases}
\end{align}

On the other hand, any IoT node that has a zenith angle in $[0,\theta_c]$ and a distance to the typical LEOSAT within the range from $a_s$ to $ (r_e + r_s)$ will be counted as an interferer (while its impact depends on its location/distance) with respect to the typical LEOSAT. Let $r_i$ be a random variable that represents the distance between the interferer node $i$ and the typical LEOSAT. Hence, its PDF is defined as

\begin{align}
\label{Eq:Interfer_Contact_Distance_PDF}
f_{R_i}(r_i)=
            \begin{cases}
                \dfrac{r_i}{\pi r_e r_s} \dfrac{1}{\sqrt{1-\left(\dfrac{r_e^2+r_s^2-r_i^2}{2r_er_s}\right)^2}}, & \\
                \hspace{4cm} a_s < r_i < r_e + r_s \\
                0, &\hspace{-3cm} otherwise,
            \end{cases}
\end{align}
Please refer to Appendices A and B in \cite{Anna1} for further details.

\subsection{Channel Model}
According to the orbital motion of LEOSATs, communication between the ground IoT node and the LEOSAT can involve two distinct types of transmissions: line of sight (LoS) and non-line of sight (NLoS) transmissions. LoS transmission occurs when there is a direct connection between the LEOSATs and the ground node. In contrast, NLoS transmission occurs when obstacles intermittently block the signal. To account for both types of transmissions, the Shadowed Rician fading model is considered to represent the channel fading between the LEOSATs and the ground nodes,
where the channel power gain between the ground node $v$ and its serving LEOSAT follows the distribution $|h_v|^2 \sim \mathcal{SR}(\Omega,b_0,m)$ \cite{Alouini_1}. Consequently, the CDF of the $\mathcal{SR}$ power gain $|h_v|^2$ is presented by

\begin{multline}
    \label{Eq:SR_CDF}
    F_{|h_v|^2}(x) =  \left(\dfrac{2b_0m}{2b_0m+\Omega}\right)^m \times \\
    \sum^\infty_{k=0} \dfrac{(m)_k}{k! \Gamma(k+1)} \left(\dfrac{\Omega}{2b_0m+\Omega}\right)^k\gamma(k+1,\frac{1}{2b_0}x),
\end{multline}
where $(m)_k$ represents the Pochhammer symbol \cite{Table_of_integrals}, $\Gamma(.)$ and $\gamma(.,.)$ are Gamma and lower incomplete Gamma functions, respectively. The term $2b_0$ is the average power of the non-line-of-sight (NLOS) multipath components, while $\Omega$ refers to the average power of the line-of-sight (LOS) component, and $m \ge 0$ is the Nakagami parameter. Consequently, $\mathcal{SR}(\Omega,b_0,m)$ denotes the parameters of the Shadowed Rician distribution. 

\subsection{Queuing Model}
 Due to the orbital motion of LEOSATs, the ground source node will not be covered by the same LEOSAT for an extended period. To ensure reliable communication, the source node first checks the availability of a typical LEOSAT before sending any updates. If the source node is within the coverage area of the LEOSAT, it establishes a connection and begins transmitting its update. 
 Moreover, the M/G/1 queuing model with server vacation assumes that arrivals follow a Poisson distribution while service times adhere to a general distribution. The server operates in two distinct states: working and vacation states. In the working state, there is at least one packet in the queue waiting to be served. On the other hand, when the server is empty, it enters a vacation state, taking a break. During this vacation, any new arrivals must wait in the queue until the server returns from its vacation and resumes service. Hence, this behavior motivates us to consider it for the LEOSAT availability due to its orbital period. Accordingly, we consider that each source node has a zero-buffer M/G/1 queue with server vacation, with no waiting space in the queue. Therefore, any newly generated update will only be processed when the server is empty. If the server is currently busy with another update, any new updates will be ignored until the ongoing update is successfully transmitted. Additionally, we account for a server vacation mechanism that is dependent on the availability of the LEOSAT. This means that the server at the source node will be in service when it is covered by the LEOSAT. If the source node is not covered, the server will enter a vacation state, and any incoming updates will be dropped until the source node is once again covered by the LEOSAT. Let $P_v$ denote the vacation probability, which refers to the probability of no LEOSAT covering the intended source node. As the LEOSATs are Poisson distributed, then the $P_{v}$ is defined as
 
 \begin{align}
\label{Eq:vacation_prob}
        P_{v} = \mathbb{P}(\text{No LEOSAT in the spherical cap})= e^{-\gamma_s {|\mathcal A_{\rm Cap}|}},
\end{align}
where ${|\mathcal A_{\rm Cap}|}$ is the area of the LEOSAT's spherical cap, which is defined as

\begin{align}
\label{Eq:A_Cap}
        {|\mathcal A_{\rm Cap}|}  &= \pi (r^2+l^2) \nonumber \\
         &= 2 \pi (r_E +a_s)^2(1-\cos{(\theta_c)}), 
\end{align}
where $r$ and $l$ are the radius and length of the spherical cap ${\mathcal A_{\rm Cap}}$ as shown in Fig. \ref{fig:System_Model}, and $\theta_c$ is the contact angle that is defined in \eqref{Eq:Contact_Angle}.

\section{Spatial Distribution of the Conditional Mean PAoI Analysis}\label{Sec:Conditional_PAoI_Analysis}

\begin{figure}[t]
\centering
\begin{tikzpicture}[y=1cm, x=1cm, thick, font=\footnotesize] 
\usetikzlibrary{arrows,decorations.pathreplacing}
\tikzset{number line/.style={}}
\tikzset{
brace_top/.style={
color=blue,
decoration={brace},
decorate
},
brace_bottom/.style={
color=blue,
decoration={brace, mirror},
decorate
}
}

\begin{scope}[xshift=0.1cm]
\tikzset{cross/.style={cross out, draw=red, minimum size=2*(#1-\pgflinewidth), inner sep=0pt, outer sep=0pt},
cross/.default={3pt}}
\draw[->] (0,0) -- (7.5,0);
\node[anchor=north] at (0,0) {$\colb{t_{k-1}}$};
\draw [blue,-{Stealth}] (-0.1,0)--(-0.1,0.75);
\node[anchor=north] at (2,0) {$\colb{t_{k-1}'}$};
\draw [blue,{Stealth}-] (2,0)--(2,0.75);
\draw [blue,-{Stealth}] (3.3,0)--(3.3,0.75);
\node[anchor=north] at (3.3,0) {$\colb{t_{k}}$};
\draw[decorate,decoration={brace, mirror,amplitude=5pt}] (0,-0.5) -- (2,-0.5)
node[anchor=north,midway,below=4pt] {$ T_{k-1}$};

\draw [blue,{Stealth}-] (6.75,0)--(6.75,0.75);
\node[anchor=north] at (6.75,0) {$\colb{t_{k}'}$};
\draw [blue,-{Stealth},dashed] (4.5,0)--(4.5,0.75);
\draw [green,-{Stealth}] (5.2,0)--(5.2,0.75);
\draw [green,-{Stealth}] (6,0)--(6,0.75);
\draw [blue,-{Stealth}, dashed] (0.8,0)--(0.8,0.75);
\draw[decorate,decoration={brace, mirror,amplitude=5pt}] (2.1,-0.5) -- (3.3,-0.5)
node[anchor=north,midway,below=4pt] { $ T^I_k $};
\draw[decorate,decoration={brace, mirror,amplitude=5pt}] (3.4,-0.5) -- (4.4,-0.5)
node[anchor=north,midway,below=4pt] { $ T_k^{W} $};
\draw[decorate,decoration={brace, mirror,amplitude=5pt}] (4.5,-0.5) -- (6.75,-0.5)
node[anchor=north,midway,below=4pt] { $ T_k^{S} $};
\end{scope}

\usetikzlibrary{arrows,decorations.pathreplacing}
\tikzset{number line/.style={}}
\tikzset{
brace_top/.style={
color=blue,
decoration={brace},
decorate
},
brace_bottom/.style={
color=blue,
decoration={brace, mirror},
decorate
}
}
\begin{scope}
\tikzset{cross/.style={cross out, draw=red, minimum size=2*(#1-\pgflinewidth), inner sep=0pt, outer sep=0pt},
cross/.default={3pt}}
\begin{axis}[ymin=0,
axis x line=bottom,
axis y line=left,
ylabel=$\Delta (t)$,xmax=17,ymax=12,
xtick={0},xticklabels={$t_k$,$t_k$},
ytick={0},yticklabels={},
extra x tick style={tick style={draw=none}}
]
\addplot [jump mark left,mark=none,draw=blue,domain=1:2,samples=2,thick]
{x+1};

\draw [dashed,draw=blue] (2,2) -- (2,3);
\draw [draw=blue,thick] (2,3) -- (3,3);
\draw [dashed,draw=blue] (3,3) -- (3,4);
\draw [draw=blue,thick] (3,4) -- (4,4);
\draw [dashed,draw=blue] (4,4) -- (4,5);
\draw [draw=blue,thick] (4,5) -- (5,5);
\draw [dashed,draw=blue] (5,5) -- (5,6);
\draw [draw=blue,thick] (5,6) -- (6,6);
\draw [dashed,draw=blue] (6,6) -- (6,2);
\draw [draw=black,thick] (6,2) -- (7,2);
\draw [dashed,draw=black] (7,2) -- (7,3);
\draw [draw=black,thick] (7,3) -- (8,3);
\draw [dashed,draw=black] (8,3) -- (8,4);
\draw [draw=black,thick] (8,4) -- (9,4);
\draw [dashed,draw=blue] (9,4) -- (9,5);
\draw [draw=blue,thick] (9,5) -- (10,5);
\draw [dashed,draw=blue] (10,5) -- (10,6);
\draw [draw=blue,thick] (10,6) -- (11,6);
\draw [dashed,draw=blue] (11,6) -- (11,7);
\draw [draw=blue,thick] (11,7) -- (12,7);
\draw [dashed,draw=blue] (12,7) -- (12,8);
\draw [draw=blue,thick] (12,8) -- (13,8);
\draw [dashed,draw=blue] (13,8) -- (13,9);
\draw [draw=blue,thick] (13,9) -- (14,9);
\draw [dashed,draw=blue] (14,9) -- (14,10);
\draw [draw=blue,thick] (14,10) -- (15,10);
\draw [draw=blue,dashed] (15,10) -- (15,11);
\draw [draw=blue,thick] (15,11) -- (16,11);
\draw [draw=blue,dashed] (16,11) -- (16,12);
\draw [draw=blue,thick] (16,12) -- (17,12);
\draw [dashed,draw=blue] (17,12) -- (17,5);
\end{axis}

\draw [decorate,decoration={ brace, mirror,amplitude=5pt}] (7.1,0) -- (7.1,2.3) 
node[anchor=east,midway,right=4pt] { $ T_k $};
\draw [decorate,decoration={ brace, mirror,amplitude=5pt}] (3.3,-1.2) -- (6.9,-1.2) 
node[anchor=north,midway,below=4pt] { $ T_k $};
\node[anchor=north] at (-0.4,1) {$T_{k-2}$};
\node[anchor=north] at (1,0) {$\colb{t^S_{k-1}}$};
\node[anchor=north] at (4.65,0) {$\colb{t^S_{k}}$};
\draw (5.3,0) node[cross,rotate=10] {};
\draw (6.1,0) node[cross,rotate=10] {};
\end{scope}
\end{tikzpicture}	
\caption{\small{AoI evolution for an update transmission. Upward and downward blue arrows show the generation and reception of the transmitted update, while dashed upward blue arrows show the start of the service times of the update. Green upward arrows show when new updates are ignored till serving the in-service one.}}
\label{fig:AoI_non_preemptive_approach}
\end{figure}
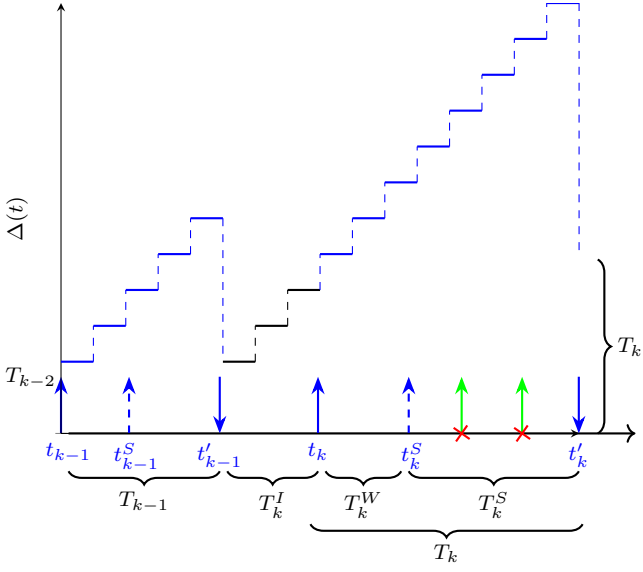

The primary objective of this work is to derive the spatial distribution of the uplink mean PAoI at the typical LEOSAT. To achieve this, we begin our analysis by defining the conditional uplink coverage probability. We then derive its b-th moment and the meta distribution. Finally, we formulate the b-th moment of the uplink mean PAoI and its spatial distribution expressions.

\begin{definition}[Conditional Uplink Coverage Probability]
Given $\Psi_{n}$ of the ground IoT nodes, the conditional uplink coverage probability is the probability that the SINR  at the typical LEOSAT exceeds a certain threshold $\gamma_{\rm th}$, which is given by
\end{definition}
\begin{equation}\label{Eq:Conditional_Uplink_Pcov} 
 \begin{aligned}
	 P^{\rm U}_{\rm cov}(\gamma_{\rm th}|\Psi_{n}) &= \mathbb{P}({{\rm{SINR}} \rm} >\gamma_{\rm th}|\Psi_{n})  \\
        &= \mathbb{P}\left(\dfrac{ {\rho_t^{\rm IoT}} \left(\dfrac{\lambda_c}{4 \pi r_o}\right)^2 |H_o|^2}{ {\mathcal I_{\rm agg}^{\rm U}}+\sigma_S^2} >\gamma_{\rm th} |\Psi_{n} \right), 
\end{aligned}
\end{equation}
where $\rho_t^{\rm IoT}$ is the transmission power of the ground IoT node, $r_o$ and $|H_o|^2$ refer to the distance and the channel power gain between the intended source node and the typical LEOSAT, respectively, the CDF of the $\mathcal{SR}$ power gain $|H_o|^2$ is expressed in \eqref{Eq:SR_CDF}, $\lambda_c = c/f_c$ denotes the wavelength of the transmitted signal, $c$ and $f_c$ represent the speed of light in free space $(c=3\times 10^8 \text{m/s})$ and the carrier frequency of the transmitted signal, respectively, $\sigma_S^2$ is the noise power at the typical LEOSAT, and ${\mathcal I_{\rm agg}^{\rm U}}$ denotes the uplink aggregated interference at the typical LEOSAT from all potential interferers within its footprint and is defined as
\begin{equation}\label{Eq:I_agg} 
 \begin{aligned}
	 {\mathcal I_{\rm agg}^{\rm U}} &= \sum_{i\in\{\Psi_n \cap \text{A(footprint) \rm}\}/o} {\rho_t^{\rm IoT}} \left(\dfrac{\lambda_c}{4 \pi r_i}\right)^2 |H_i|^2, \nonumber \\
\end{aligned}
\end{equation}
where $r_i$ and $|H_i|^2$ refer to the distance and the channel power gain between the interferer node $i$ and the typical LEOSAT, respectively.

\subsection{Mean PAoI}
AoI tracks the elapsed time between the generation of an update at the source node and its successful reception at the destination. We define $\Delta(t_k') = t_k' - t_k$ as the AoI for the transmitted update at the time $t_k'$, with $t_k$ denoting the generation time of the $k$th update. The behavior of $\Delta(t_k')$ exhibits a clear staircase pattern: it consistently increases over time, then drops down to the service time of the update upon its arrival at the destination node, as shown in Fig. \ref{fig:AoI_non_preemptive_approach}. We assume the IoT source node follows a non-preemptive transmission approach, where it ignores any new upcoming update till completing the transmission of the in-service one. As illustrated in Fig. \ref{fig:AoI_non_preemptive_approach}, the updates that are generated during the service time of the $k$th update are ignored (illustrated by the green upward arrows) till the reception of the  $k$th update at the typical LEOSAT.

In this work, we consider a time-slotted system, where updates are generated and received at the beginning of each time slot. As such, the minimum achievable AoI within this system is definitely one time slot. PAoI is defined as the maximum AoI in the network before the reception of the update. To determine the PAoI, we assume the updates are generated according to a Poisson distribution with parameter $\lambda_a$. Consequently, the interarrival times between consecutive updates follow an exponential distribution with the same parameter $\lambda_a$. Let $T^I_k$ denote the interarrival time after the last successfully received update till the $k$th update is generated at the source node. $T_k^S$ and $T_k^W$ represent the service and waiting to be served times of the $k$th update, respectively. Both $T_k^S$ and $T_k^W$ follow a geometric distribution with parameter $P^{\rm U}_{\rm cov}(\gamma_{\rm th})$ which is defined in \eqref{Eq:Conditional_Uplink_Pcov} and $(1-P_{v})$ which is defined in \eqref{Eq:vacation_prob}, respectively. Therefore, $T_k = T_k^S + T_k^W$ denotes the total time of the $k$th update from its generation till its successful reception. Let $A_k^U$ denote the uplink PAoI at the typical LEOSAT, which is defined as follows

\begin{equation}\label{Eq:thm_Non_PAoI} 
  \begin{aligned}
      A_k &= T_{\rm k-1} + T^I_k + T_k + T^{\rm U}_P, 
  \end{aligned}
\end{equation}
where $T^{\rm U}_P = \dfrac{1}{c}\int_{a_s}^{ {\rm r_{\rm max}}} r_o f_{R_o}(r_o) dr_o$ is the uplink propagation time, which is a random variable that depends on the distance $R_o$ between the source node and the typical LEOSAT.

\begin{definition}[Conditional Uplink Mean PAoI] \label{thm:Cond_peak_AoI}
Conditioned on the HPPP $\Psi_n$ for the ground nodes, the conditional uplink mean PAoI $\bar{A}_k$
is defined as: 

\begin{align}\label{Eq:Cond_PAoI_Mean}
      &\bar{A}_k = \mathbb{E}[A_k(\gamma_{\rm th}|\Psi_n)] \nonumber \\
      &= \mathbb{E}[(T_{\rm k-1} +T_k + T^I_k + T^{\rm U}_P) |\Psi_n],\nonumber \\
  \end{align}
where the expectation is over the channel fading for a given location of the IoT nodes.
\end{definition}
Due to the PPP being stationary and ergodic, the spatial average of the PAoI $\bar{A}_k$ is the ensemble mean across different network realizations. Moreover, from the spatial domain perspective, the conditional coverage probability is a random variable that varies in each network realization \cite{ElSawy1}. The service time $T_k^S$ of the $k$th update follows a geometric distribution with a parameter $P^{\rm U}_{\rm cov}(\gamma_{\rm th})$, which depends on the network realization. Consequently, to determine the expectation of the conditional mean PAoI in \eqref{Eq:Cond_PAoI_Mean}, we need to derive the b-th moment of the coverage probability first.

\subsection{b-th Moment of Conditional Coverage Probability}
Conditioned on a given realization of $\Psi_n$, the source of randomness becomes only the channel power gain of the interferer nodes $|H_i|^2$, where all potential interferers experience shadowed Rician fading. Therefore, the b-th moment expression of the conditional coverage probability is given by the following theorem. 
\begin{theorem}[b-th Moment of Conditional Coverage Probability]
\label{thm:bth_Moment_of_Conditional_Coverage_Probability}
Conditioned on the ground IoT node locations $\Psi_n$, the b-th moment of the conditional coverage probability at the typical LEOSAT is approximated as: 
\begin{align}
\label{Eq:Uplink_Coverage_Probability_bth_Moments_Final}
&M^b(\gamma_{\rm th}|\psi_n) =
\int_{a_s}^{ r_{\rm max}} \biggl[  \sum_{\substack{b_1+b_2+\cdots +b_k=b \\ b_1,b_2,\cdots,b_k \geq 0}} \binom{b}{b_1,b_2,\cdots, b_k}  \nonumber \\
    &\prod_{j=1}^k \biggl(\binom{k}{j} (-1)^{j+1} \biggr)^{b_j} \exp\biggl(\dfrac{-j\mu_1 \gamma_{\rm th}\sigma_S^2 r{_o^{\rm U}}^2 \sum_{j=1}^k j b_j}{m_h{\rho_t^{\rm IoT}} \nu}\biggr) \times \nonumber \\
    &\prod_{j=1}^k \biggl({\mathcal{L}_{\mathcal I_{\rm agg}^{\rm U}}} \biggl(\dfrac{-j\mu_1 \gamma_{\rm th} r{_o^{\rm U}}^2}{m_h{\rho_t^{\rm IoT}} \nu}\biggr)\biggr)^{b_j}\biggr] f_{R_o}(r_o) dr_o,
\end{align}
where $\mathcal{L}_{\mathcal I_{\rm agg}^{\rm U}}$ is the Laplace transform for $\mathcal I_{\rm agg}^{\rm U}$, which is defined in (12), $\nu$ is $(\lambda_c/4 \pi)^2$, $k=\frac{m(2b_0+\Omega)^2}{4b_0^2m+4b_0m\Omega+\Omega^2}$, and $m_h = \frac{4b_0^2m+4b_0m\Omega+\Omega^2}{m(2b_0+\Omega)^2}$ are the shape and scale parameters of Gamma distribution of $|H_o|^2$, and $\mu_1$ is the upper tight bound parameter of the of Gamma function, which is defined as
\begin{align}
    \label{Eq:New_tight_bounds} 
    \mu_1 =
        \begin{cases} 
            (k!)^{\dfrac{-1}{k}}  & k > 1 \\
            1 & k < 1 .
        \end{cases} 
\end{align}
\end{theorem}
\begin{IEEEproof}
See Appendix~\ref{app:bth_Uplink_Pcov}.
\end{IEEEproof}

To calculate $\mathcal{L}_{\mathcal I_{\rm agg}^{\rm U}}$, we condition on the random distance of the intended source node $(R_o = r_o)$ to the typical LEOSAT. Hence, $\mathcal{L}_{\mathcal I_{\rm agg}^{\rm U}}$ is given by the following lemma. 
\begin{lemma}\label{lemma:Uplink_Iagg}
 The Laplace transform of the aggregate interference is
\begin{align}
\label{Eq:Iagg_Final}
 {\mathcal{L}_{\mathcal I_{\rm agg}^{\rm U}}(s)} &= \biggl[\int_{a_s}^{\rm {r_{\rm max}}} \left(1+\left(\dfrac{s{\rho_t^{\rm IoT}}\nu m_h}{r_i^2}\right)\right)^{-k} f_{R_i}(r_i) dr_i \nonumber \\ 
        &+ \int_{\rm {r_{\rm max}}}^{(r_s+r_e)} f_{R_i}(r_i) dr_i\biggr]^{N_s-1}.
\end{align}

\begin{IEEEproof} 
See Appendix~\ref{app:Uplink_Iagg}.
\end{IEEEproof}
\end{lemma}

Consequently, the b-th moment of the conditional mean PAoI is given by the following theorem.
\begin{theorem}[b-th Moment of Conditional Mean PAoI]
\label{thm:bth_Moment_of_Conditional_Mean_PAoI}
Based on the conditional mean PAoI in \eqref{Eq:Cond_PAoI_Mean} and the expression of $M^b(\gamma_{\rm th}|\psi_n)$ in \eqref{Eq:Uplink_Coverage_Probability_bth_Moments_Final}, the b-th moment of the conditional mean PAoI is approximated as:
\begin{align}
\label{Eq:bth_Moment_of_Conditional_Mean_PAoI}
&M^b(\bar{A}_k) =
\sum^b_{x=0} \binom{b}{x} \biggl(\dfrac{2}{1-p_{v}}+\dfrac{1}{\lambda_a}\biggr)^{b-x} \times \nonumber\\
&\Bigg[\sum^x_{y=0}\binom{x}{y} \left(\dfrac{2}{M^b(\gamma_{\rm th}|\psi_n)}\right)^x \biggl(\dfrac{1}{c}\int_{a_s}^{ {\rm r_{\rm max}}} r_o f_{R_o}(r_o) dr_o \biggr)^{x-y}\Bigg].
\end{align}
\end{theorem}
\begin{IEEEproof}
See Appendix~\ref{app:bth_Moment_of_Conditional_Mean_PAoI}.
\end{IEEEproof}
Now we can compute the mean PAoI in the system by substituting for $b=1$ in the result in Theorem \ref{thm:bth_Moment_of_Conditional_Mean_PAoI}.\\

To determine the spatial distribution of the conditional mean PAoI, we first need to derive the distribution of the conditional coverage probability. Consequently, the meta distribution concept is introduced to capture the variation of the coverage probability across the different network realizations, which is defined next.

\begin{definition}[Coverage Probability Meta Distribution]
The coverage probability meta distribution at the typical LEOSAT $\bar{F}_{\rm Cov}(\gamma_{\rm th}, \tau)$ is the complementary cumulative distribution function (CCDF) of the random variable $P^{\rm U}_{\rm cov}(\gamma_{\rm th}|\Psi_{n})$ defined for a reliability threshold $\tau \in [0,1]$.
\end{definition}
As discussed in \cite{Martin1}, the Beta approximation of the SIR meta distribution aligns perfectly with the exact distribution rather than Markov, Chebyshev, and Paley-Zygmund bounds. This strong correlation establishes the Beta approximation as a widely used method for determining the meta distribution of any random variable. Hence, we adopt the Beta distribution approximation to approximate the meta-distribution of the coverage probability at the typical LEOSAT.
The probability distribution function (pdf) of a Beta-distributed random variable $X$ with mean $\mu$ and variance $\sigma^2$ is given by

\begin{align}
\label{Eq:Beta_Distribution_pdf}
&f_{X}(x) = \dfrac{x^{\dfrac{\mu(\zeta+1)-1}{1-\mu}} (1-x)^{\zeta-1}}{B\left(\dfrac{\mu \zeta}{(1-\mu)},\zeta\right)},
\end{align}
where $\zeta = \dfrac{\mu (1-\mu)^2}{\sigma^2} -(1-\mu)$ and $B(\eta,\zeta)=\dfrac{\Gamma(\eta)+\Gamma(\zeta)}{\Gamma(\eta+\zeta)}$ is the beta function for $\eta >0$ and  $\zeta>0$ \cite{Martin1} and $\Gamma(s) = \int_0^\infty e^{-x} x^{s-1} dx$ is the Gamma function. \\

Therefore, the approximated meta-distribution of the coverage probability is $\bar{F}_{\rm Cov}(\gamma_{\rm th}, \tau)=\mathbb{P}(\mathbb{P}({{\rm{SINR}} \rm} >\gamma_{\rm th}|\Psi_{n}) > \tau) = \mathbb{P}(P^{\rm U}_{\rm cov}(\gamma_{\rm th}|\psi_n) > \tau)$, as follows:

\begin{equation}\label{Eq:Meta_Distribution_2}
\begin{aligned}
    &\bar{F}_{\rm Cov}(\gamma_{\rm th}, \tau) = 1 - \mathcal{I}_\tau \left(\alpha,\beta\right), 
  \end{aligned}
\end{equation}
where $\mathcal{I}_\tau (\alpha,\beta) = \int_0^\tau t^{\alpha-1} (1-t)^{\beta-1} dt$ is the regularized incomplete beta function \cite{ElSawy1}, $\tau \in [0,1]$ denotes the reliability threshold, $\alpha$ and $\beta$ are defined as follows,

\begin{equation}\label{Eq:Alpha}
    \alpha = \dfrac{M^1(\gamma_{\rm th}|\psi_n)(M^1(\gamma_{\rm th}|\psi_n)-M^2(\gamma_{\rm th}|\psi_n))}{M^2(\gamma_{\rm th}|\psi_n)-(M^1(\gamma_{\rm th}|\psi_n))^2}  \nonumber
\end{equation} 
\vspace{0.3mm}
\begin{equation}\label{Eq:Beta}
    \beta = \dfrac{(1-M^1(\gamma_{\rm th}|\psi_n))(M^1(\gamma_{\rm th}|\psi_n)-M^2(\gamma_{\rm th}|\psi_n))}{M^2(\gamma_{\rm th}|\psi_n)-(M^1(\gamma_{\rm th}|\psi_n))^2},
\end{equation} 
where $M^1(\gamma_{\rm th}|\psi_n)$ and $M^2(\gamma_{\rm th}|\psi_n)$ are the first and second moments of $P^{\rm U}_{\rm cov}(\gamma_{\rm th}|\psi_n)$, respectively, and are determined from \eqref{Eq:Uplink_Coverage_Probability_bth_Moments_Final}.

Using \eqref{Eq:Meta_Distribution_2}, we can formulate the spatial distribution of the conditional mean PAoI, which is defined as the following definition.
\begin{definition}[Spatial Distribution of Conditional Mean PAoI] \label{def:Downlink_Coverage_Prob}
The spatial distribution of the conditional mean PAoI represents the Cumulative Distribution Function (CDF) of the conditional mean PAoI, constrained by a time deadline $t$. 
\begin{align}
\label{def:AoI_Spatial_Distribution}
F_{\bar{A}_k}(\gamma_{\rm th}, t)  &= \mathbb{P}( \bar{A}_k \leq t).
\end{align}
\end{definition}
Hence, the approximated spatial distribution of the conditional mean PAoI is given by the following theorem.
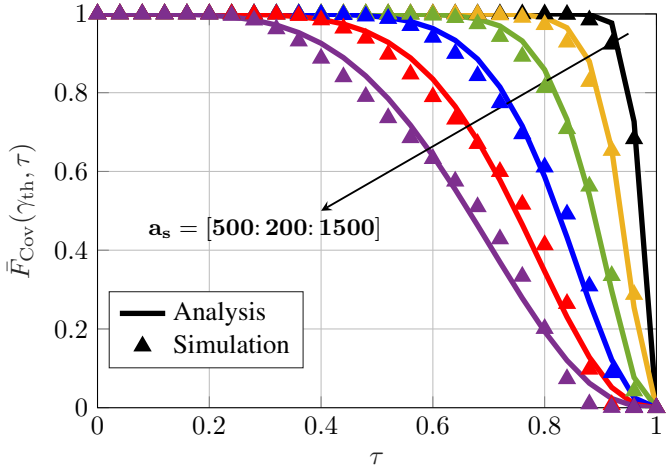
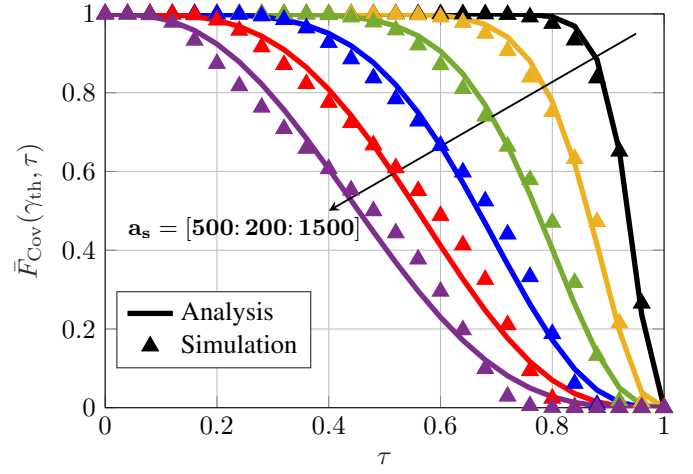
\begin{figure*}[!t]
  \centering
  \begin{subfigure}[t]{0.5\textwidth}
      \centering
%
%
\definecolor{mycolor1}{rgb}{0.92941,0.69412,0.12549}%
\definecolor{mycolor2}{rgb}{0.46667,0.67451,0.18824}%
\definecolor{mycolor3}{rgb}{0.49412,0.18431,0.55686}%
\begin{tikzpicture}

\begin{axis}[%
width=2.911in,
height=2.051in,
at={(0in,0in)},
scale only axis,
xmin=0,
xmax=1,
xlabel style={font=\color{white!15!black}},
xlabel={$\tau$},
ymin=0,
ymax=1,
ylabel style={font=\color{white!15!black}},
ylabel={$\bar{F}_{\rm Cov}(\gamma_{\rm th}, \tau)$},
axis background/.style={fill=white},
xmajorgrids,
xminorgrids,
ymajorgrids,
yminorgrids,
legend style={ at={(0.02,0.1)}, anchor=south west, legend columns=1,legend cell align=left, align=left, draw=white!15!black}
]
\addplot [color=black, line width=2.0pt]
  table[row sep=crcr]{%
0	1\\
0.04	1\\
0.08	1\\
0.12	1\\
0.16	1\\
0.2	1\\
0.24	1\\
0.28	1\\
0.32	1\\
0.36	1\\
0.4	1\\
0.44	1\\
0.48	1\\
0.52	1\\
0.56	1\\
0.6	0.999999999999999\\
0.64	0.999999999999904\\
0.68	0.999999999990811\\
0.72	0.999999999345766\\
0.76	0.999999964444056\\
0.8	0.999998502968675\\
0.84	0.999950924599289\\
0.88	0.998760728437382\\
0.92	0.976934878527739\\
0.96	0.726839628353771\\
1	0\\
};
\addlegendentry{Analysis}

\addplot [only marks,color=black, line width=2.0pt, draw=none, mark=triangle, mark options={solid, black}]
  table[row sep=crcr]{%
0	1\\
0.04	1\\
0.08	1\\
0.12	1\\
0.16	1\\
0.2	1\\
0.24	1\\
0.28	1\\
0.32	1\\
0.36	1\\
0.4	1\\
0.44	1\\
0.48	1\\
0.52	1\\
0.56	1\\
0.6	1\\
0.64	1\\
0.68	1\\
0.72	1\\
0.76	0.9999\\
0.8	0.9998\\
0.84	0.9986\\
0.88	0.9844\\
0.92	0.923\\
0.96	0.6828\\
1	0\\
};
\addlegendentry{Simulation}

\addplot [color=mycolor1, line width=2.0pt, forget plot]
  table[row sep=crcr]{%
0	1\\
0.04	1\\
0.08	1\\
0.12	1\\
0.16	1\\
0.2	1\\
0.24	1\\
0.28	1\\
0.32	0.999999999999999\\
0.36	0.999999999999951\\
0.4	0.999999999998301\\
0.44	0.999999999958533\\
0.48	0.999999999246526\\
0.52	0.99999998933715\\
0.56	0.999999878365176\\
0.6	0.999998851370652\\
0.64	0.999990833953855\\
0.68	0.999937212095448\\
0.72	0.99962657607844\\
0.76	0.998057896516388\\
0.8	0.991146148758478\\
0.84	0.96475957868707\\
0.88	0.879297780128434\\
0.92	0.656542768576017\\
0.96	0.25462242658565\\
1	0\\
};
\addplot [only marks,color=mycolor1, line width=2.0pt, draw=none, mark=triangle, mark options={solid, mycolor1}, forget plot]
  table[row sep=crcr]{%
0	1\\
0.04	1\\
0.08	1\\
0.12	1\\
0.16	1\\
0.2	1\\
0.24	1\\
0.28	1\\
0.32	1\\
0.36	1\\
0.4	1\\
0.44	1\\
0.48	1\\
0.52	1\\
0.56	1\\
0.6	1\\
0.64	0.9999\\
0.68	0.9995\\
0.72	0.9973\\
0.76	0.9911\\
0.8	0.9726\\
0.84	0.9282\\
0.88	0.8282\\
0.92	0.6531\\
0.96	0.2874\\
1	0\\
};
\addplot [color=mycolor2, line width=2.0pt, forget plot]
  table[row sep=crcr]{%
0	1\\
0.04	1\\
0.08	1\\
0.12	1\\
0.16	0.99999999999992\\
0.2	0.999999999995028\\
0.24	0.999999999857995\\
0.28	0.999999997623337\\
0.32	0.999999973141259\\
0.36	0.999999775472262\\
0.4	0.999998522728973\\
0.44	0.999992006255141\\
0.48	0.999963263379678\\
0.52	0.999853126720947\\
0.56	0.999479771534745\\
0.6	0.998344687384015\\
0.64	0.995218693029307\\
0.68	0.9873670703579\\
0.72	0.969313760756541\\
0.76	0.931304147438089\\
0.8	0.858341721458536\\
0.84	0.732102682398469\\
0.88	0.540069565191258\\
0.92	0.29711367267643\\
0.96	0.0768377007945549\\
1	0\\
};
\addplot [only marks,color=mycolor2, line width=2.0pt, draw=none, mark=triangle, mark options={solid, mycolor2}, forget plot]
  table[row sep=crcr]{%
0	1\\
0.04	1\\
0.08	1\\
0.12	1\\
0.16	1\\
0.2	1\\
0.24	1\\
0.28	1\\
0.32	1\\
0.36	1\\
0.4	1\\
0.44	1\\
0.48	1\\
0.52	0.9997\\
0.56	0.9991\\
0.6	0.9966\\
0.64	0.9905\\
0.68	0.9748\\
0.72	0.9424\\
0.76	0.8921\\
0.8	0.8128\\
0.84	0.7084\\
0.88	0.5624\\
0.92	0.3353\\
0.96	0.0424\\
1	0\\
};
\addplot [color=blue, line width=2.0pt, forget plot]
  table[row sep=crcr]{%
0	1\\
0.04	0.999999999999995\\
0.08	0.9999999999879\\
0.12	0.999999998825982\\
0.16	0.999999970530096\\
0.2	0.999999647910882\\
0.24	0.99999737365201\\
0.28	0.999985865849957\\
0.32	0.999940185374603\\
0.36	0.99978961877351\\
0.4	0.999361453173243\\
0.44	0.998282460742794\\
0.48	0.995825968027213\\
0.52	0.99070196485336\\
0.56	0.980808240672628\\
0.6	0.962994633095226\\
0.64	0.932940420372599\\
0.68	0.885302371383699\\
0.72	0.814345622116408\\
0.76	0.715295612514436\\
0.8	0.586598750992685\\
0.84	0.433065649614473\\
0.88	0.269333506809998\\
0.92	0.121896412163972\\
0.96	0.0252758937560621\\
1	0\\
};
\addplot [only marks,color=blue, line width=2.0pt, draw=none, mark=triangle, mark options={solid, blue}, forget plot]
  table[row sep=crcr]{%
0	1\\
0.04	1\\
0.08	1\\
0.12	1\\
0.16	1\\
0.2	1\\
0.24	1\\
0.28	1\\
0.32	1\\
0.36	0.9999\\
0.4	0.9991\\
0.44	0.998\\
0.48	0.9944\\
0.52	0.9881\\
0.56	0.9702\\
0.6	0.9404\\
0.64	0.8994\\
0.68	0.843\\
0.72	0.7742\\
0.76	0.6954\\
0.8	0.611\\
0.84	0.4915\\
0.88	0.3089\\
0.92	0.0888\\
0.96	0.0018\\
1	0\\
};
\addplot [color=red, line width=2.0pt, forget plot]
  table[row sep=crcr]{%
0	1\\
0.04	0.999999999058343\\
0.08	0.999999844101558\\
0.12	0.999996992665286\\
0.16	0.999975980215\\
0.2	0.999881830326316\\
0.24	0.999572667474167\\
0.28	0.998751929137619\\
0.32	0.99688635088996\\
0.36	0.993121523083818\\
0.4	0.986211697462827\\
0.44	0.974484139596998\\
0.48	0.955860183833165\\
0.52	0.927954742675867\\
0.56	0.888272794289328\\
0.6	0.834514669778449\\
0.64	0.764990979132928\\
0.68	0.679131681879313\\
0.72	0.578050603605692\\
0.76	0.465094277225425\\
0.8	0.346258390047794\\
0.84	0.230288904363961\\
0.88	0.128182092069706\\
0.92	0.0516170799749338\\
0.96	0.00945658985168518\\
1	0\\
};
\addplot [only marks,color=red, line width=2.0pt, draw=none, mark=triangle, mark options={solid, red}, forget plot]
  table[row sep=crcr]{%
0	1\\
0.04	1\\
0.08	1\\
0.12	1\\
0.16	1\\
0.2	1\\
0.24	1\\
0.28	0.9996\\
0.32	0.999\\
0.36	0.9957\\
0.4	0.9846\\
0.44	0.9637\\
0.48	0.938\\
0.52	0.8973\\
0.56	0.8471\\
0.6	0.7896\\
0.64	0.7326\\
0.68	0.6706\\
0.72	0.5995\\
0.76	0.5164\\
0.8	0.4134\\
0.84	0.2644\\
0.88	0.0975\\
0.92	0.0074\\
0.96	0\\
1	0\\
};
\addplot [color=mycolor3, line width=2.0pt, forget plot]
  table[row sep=crcr]{%
0	1\\
0.04	0.999998834538107\\
0.08	0.999962784242995\\
0.12	0.999725377426662\\
0.16	0.998889185062106\\
0.2	0.996771900952308\\
0.24	0.992397767928739\\
0.28	0.984530757649493\\
0.32	0.971744498620952\\
0.36	0.952523011978679\\
0.4	0.925385553056098\\
0.44	0.88902815175865\\
0.48	0.842473731339548\\
0.52	0.785221913656927\\
0.56	0.717388738888361\\
0.6	0.63982547596707\\
0.64	0.554204386999793\\
0.68	0.463057595331436\\
0.72	0.369752860722387\\
0.76	0.278386672586989\\
0.8	0.193569837008053\\
0.84	0.120071971845656\\
0.88	0.0622750168974642\\
0.92	0.0233506307838824\\
0.96	0.0039800865191788\\
1	0\\
};
\addplot [only marks,color=mycolor3, line width=2.0pt, draw=none, mark=triangle, mark options={solid, mycolor3}, forget plot]
  table[row sep=crcr]{%
0	1\\
0.04	1\\
0.08	1\\
0.12	1\\
0.16	1\\
0.2	0.9993\\
0.24	0.9959\\
0.28	0.9842\\
0.32	0.9619\\
0.36	0.931\\
0.4	0.8872\\
0.44	0.8401\\
0.48	0.7901\\
0.52	0.7361\\
0.56	0.6854\\
0.6	0.6324\\
0.64	0.5751\\
0.68	0.5101\\
0.72	0.4285\\
0.76	0.3334\\
0.8	0.2009\\
0.84	0.0734\\
0.88	0.0095\\
0.92	0.0004\\
0.96	0\\
1	0\\
};
\addplot[color=black,line width=0.25mm, ->,>=stealth
]
table[row sep=crcr]{%
0.95  0.95 \\
0.4   0.5\\
};
\node[above] at (axis cs:0.3,0.4) {\small{$\mathbf{ a_s=[500 \colon 200 \colon 1500]}$}};
\end{axis}

\begin{axis}[%
width=2.911in,
height=2.051in,
at={(0in,0in)},
scale only axis,
xmin=0,
xmax=1,
ymin=0,
ymax=1,
axis line style={draw=none},
ticks=none,
axis x line*=bottom,
axis y line*=left,
legend style={legend cell align=left, align=left, draw=white!15!black}
]
\end{axis}
\end{tikzpicture}%
\caption{\small{Number of LEOSATs ($N_s =900$).}}
\label{Fig:P_Cov_Meta_Distribution_Ns_900}
  \end{subfigure}%
  ~ 
    \begin{subfigure}[t]{0.5\textwidth}
    \centering
%
%
\definecolor{mycolor1}{rgb}{0.92941,0.69412,0.12549}%
\definecolor{mycolor2}{rgb}{0.46667,0.67451,0.18824}%
\definecolor{mycolor3}{rgb}{0.49412,0.18431,0.55686}%
\begin{tikzpicture}

\begin{axis}[%
width=2.911in,
height=2.051in,
at={(0in,0in)},
scale only axis,
xmin=0,
xmax=1,
xlabel style={font=\color{white!15!black}},
xlabel={$\tau$},
ymin=0,
ymax=1,
ylabel style={font=\color{white!15!black}},
ylabel={$\bar{F}_{\rm Cov}(\gamma_{\rm th}, \tau)$},
axis background/.style={fill=white},
xmajorgrids,
xminorgrids,
ymajorgrids,
yminorgrids,
legend style={ at={(0.02,0.1)}, anchor=south west, legend columns=1,legend cell align=left, align=left, draw=white!15!black}
]
\addplot [color=black, line width=2.0pt]
  table[row sep=crcr]{%
0	1\\
0.04	1\\
0.08	1\\
0.12	1\\
0.16	1\\
0.2	1\\
0.24	1\\
0.28	1\\
0.32	1\\
0.36	0.999999999999996\\
0.4	0.999999999999789\\
0.44	0.999999999993217\\
0.48	0.999999999841796\\
0.52	0.999999997191243\\
0.56	0.999999960597716\\
0.6	0.999999550466516\\
0.64	0.999995735273162\\
0.68	0.999965785119729\\
0.72	0.999765058654523\\
0.76	0.998609108678363\\
0.8	0.992887434786162\\
0.84	0.968757485145124\\
0.88	0.884190670673706\\
0.92	0.652586925683841\\
0.96	0.235779342510286\\
1	0\\
};
\addlegendentry{Analysis}

\addplot [only marks,color=black, line width=2.0pt, draw=none, mark=triangle, mark options={solid, black}]
  table[row sep=crcr]{%
0	1\\
0.04	1\\
0.08	1\\
0.12	1\\
0.16	1\\
0.2	1\\
0.24	1\\
0.28	1\\
0.32	1\\
0.36	1\\
0.4	1\\
0.44	1\\
0.48	1\\
0.52	1\\
0.56	1\\
0.6	1\\
0.64	1\\
0.68	0.9997\\
0.72	0.9982\\
0.76	0.992\\
0.8	0.9756\\
0.84	0.9334\\
0.88	0.8371\\
0.92	0.6512\\
0.96	0.2651\\
1	0\\
};
\addlegendentry{Simulation}

\addplot [color=mycolor1, line width=2.0pt, forget plot]
  table[row sep=crcr]{%
0	1\\
0.04	1\\
0.08	1\\
0.12	0.999999999999993\\
0.16	0.99999999999894\\
0.2	0.999999999948928\\
0.24	0.999999998815082\\
0.28	0.99999998342572\\
0.32	0.999999840215067\\
0.36	0.999998842934252\\
0.4	0.999993327526761\\
0.44	0.999968063688121\\
0.48	0.999869247753561\\
0.52	0.999531691630965\\
0.56	0.998507617974132\\
0.6	0.995714206741973\\
0.64	0.988803306173198\\
0.68	0.973211493666571\\
0.72	0.941062988963808\\
0.76	0.880584681251365\\
0.8	0.777498871042336\\
0.84	0.620728180888004\\
0.88	0.414599930485844\\
0.92	0.195710890398874\\
0.96	0.0385761321429\\
1	0\\
};
\addplot [only marks,color=mycolor1, line width=2.0pt, draw=none, mark=triangle, mark options={solid, mycolor1}, forget plot]
  table[row sep=crcr]{%
0	1\\
0.04	1\\
0.08	1\\
0.12	1\\
0.16	1\\
0.2	1\\
0.24	1\\
0.28	1\\
0.32	1\\
0.36	1\\
0.4	1\\
0.44	0.9999\\
0.48	0.9998\\
0.52	0.9993\\
0.56	0.9971\\
0.6	0.9918\\
0.64	0.978\\
0.68	0.9503\\
0.72	0.9066\\
0.76	0.8392\\
0.8	0.7517\\
0.84	0.6326\\
0.88	0.4724\\
0.92	0.2152\\
0.96	0.0143\\
1	0\\
};
\addplot [color=mycolor2, line width=2.0pt, forget plot]
  table[row sep=crcr]{%
0	1\\
0.04	0.999999999999416\\
0.08	0.999999999393717\\
0.12	0.999999966107856\\
0.16	0.99999942835788\\
0.2	0.999995009446944\\
0.24	0.999971326056114\\
0.28	0.999876776368812\\
0.32	0.99957259937032\\
0.36	0.998743562554434\\
0.4	0.996764059989683\\
0.44	0.992525872568399\\
0.48	0.984253557900486\\
0.52	0.969358594514847\\
0.56	0.944409674887014\\
0.6	0.905316684797774\\
0.64	0.847829586119327\\
0.68	0.768427321113868\\
0.72	0.66560029105847\\
0.76	0.541396312009827\\
0.8	0.402889668542708\\
0.84	0.262939753422913\\
0.88	0.139245986798926\\
0.92	0.050354648206272\\
0.96	0.00718949346285402\\
1	0\\
};
\addplot [only marks,color=mycolor2, line width=2.0pt, draw=none, mark=triangle, mark options={solid, mycolor2}, forget plot]
  table[row sep=crcr]{%
0	1\\
0.04	1\\
0.08	1\\
0.12	1\\
0.16	1\\
0.2	1\\
0.24	1\\
0.28	1\\
0.32	1\\
0.36	0.9995\\
0.4	0.9977\\
0.44	0.992\\
0.48	0.9791\\
0.52	0.9592\\
0.56	0.9209\\
0.6	0.8707\\
0.64	0.8101\\
0.68	0.7402\\
0.72	0.6636\\
0.76	0.5789\\
0.8	0.4706\\
0.84	0.3177\\
0.88	0.1325\\
0.92	0.0176\\
0.96	0\\
1	0\\
};
\addplot [color=blue, line width=2.0pt, forget plot]
  table[row sep=crcr]{%
0	1\\
0.04	0.999999953636986\\
0.08	0.999996346352235\\
0.12	0.999954782909336\\
0.16	0.999738299501545\\
0.2	0.999003059275532\\
0.24	0.997089681394262\\
0.28	0.992940469328409\\
0.32	0.985068985038367\\
0.36	0.971601351576852\\
0.4	0.950399571173538\\
0.44	0.919266804682452\\
0.48	0.876223576213284\\
0.52	0.819832773086067\\
0.56	0.749540620526981\\
0.6	0.665991071048748\\
0.64	0.571262894633329\\
0.68	0.468972980292701\\
0.72	0.364187029362197\\
0.76	0.263081486264205\\
0.8	0.172310659984426\\
0.84	0.0980547863008923\\
0.88	0.0447665594355888\\
0.92	0.0137143556465015\\
0.96	0.00159455518344664\\
1	0\\
};
\addplot [only marks,color=blue, line width=2.0pt, draw=none, mark=triangle, mark options={solid, blue}, forget plot]
  table[row sep=crcr]{%
0	1\\
0.04	1\\
0.08	1\\
0.12	1\\
0.16	0.9999\\
0.2	0.9999\\
0.24	0.9988\\
0.28	0.9946\\
0.32	0.9854\\
0.36	0.9641\\
0.4	0.9272\\
0.44	0.8855\\
0.48	0.8372\\
0.52	0.7841\\
0.56	0.7276\\
0.6	0.6656\\
0.64	0.5982\\
0.68	0.5248\\
0.72	0.4407\\
0.76	0.3325\\
0.8	0.1878\\
0.84	0.0614\\
0.88	0.008\\
0.92	0.0001\\
0.96	0\\
1	0\\
};
\addplot [color=red, line width=2.0pt, forget plot]
  table[row sep=crcr]{%
0	1\\
0.04	0.999979235651173\\
0.08	0.999611493046268\\
0.12	0.997925133757199\\
0.16	0.99337994948776\\
0.2	0.984096977766281\\
0.24	0.968118129922065\\
0.28	0.943655870573519\\
0.32	0.909311632678197\\
0.36	0.86425005095004\\
0.4	0.808321048166875\\
0.44	0.742125211484657\\
0.48	0.667020507994239\\
0.52	0.585070580423124\\
0.56	0.498936848277844\\
0.6	0.411718563507798\\
0.64	0.326746957149892\\
0.68	0.247341795679538\\
0.72	0.176541211930509\\
0.76	0.116818838619587\\
0.8	0.0698064794507991\\
0.84	0.0360466033446827\\
0.88	0.0148085256370051\\
0.92	0.00401936430664585\\
0.96	0.000398318525359098\\
1	0\\
};
\addplot [only marks,color=red, line width=2.0pt, draw=none, mark=triangle, mark options={solid, red}, forget plot]
  table[row sep=crcr]{%
0	1\\
0.04	1\\
0.08	1\\
0.12	0.9998\\
0.16	0.9976\\
0.2	0.9841\\
0.24	0.9581\\
0.28	0.9161\\
0.32	0.871\\
0.36	0.8227\\
0.4	0.775\\
0.44	0.7237\\
0.48	0.6671\\
0.52	0.6091\\
0.56	0.5511\\
0.6	0.4878\\
0.64	0.4131\\
0.68	0.3255\\
0.72	0.2104\\
0.76	0.0937\\
0.8	0.0236\\
0.84	0.0022\\
0.88	0\\
0.92	0\\
0.96	0\\
1	0\\
};
\addplot [color=mycolor3, line width=2.0pt, forget plot]
  table[row sep=crcr]{%
0	1\\
0.04	0.999241891724432\\
0.08	0.994091648040271\\
0.12	0.981075719524947\\
0.16	0.957946182777485\\
0.2	0.923605930064118\\
0.24	0.87796522358487\\
0.28	0.821776996843764\\
0.32	0.756466091114384\\
0.36	0.683959700829795\\
0.4	0.606523316593833\\
0.44	0.526605087224899\\
0.48	0.446690828412974\\
0.52	0.369171542593674\\
0.56	0.296225139484948\\
0.6	0.229713996028135\\
0.64	0.17110004276767\\
0.68	0.1213792084379\\
0.72	0.0810373129410283\\
0.76	0.0500299135901209\\
0.8	0.0277892662224851\\
0.84	0.0132626351137051\\
0.88	0.00498805273141256\\
0.92	0.00121723724466427\\
0.96	0.000103631190110143\\
1	0\\
};
\addplot [only marks,color=mycolor3, line width=2.0pt, draw=none, mark=triangle, mark options={solid, mycolor3}, forget plot]
  table[row sep=crcr]{%
0	1\\
0.04	1\\
0.08	0.9977\\
0.12	0.9784\\
0.16	0.9327\\
0.2	0.8739\\
0.24	0.8172\\
0.28	0.7628\\
0.32	0.7084\\
0.36	0.6583\\
0.4	0.6066\\
0.44	0.5526\\
0.48	0.4998\\
0.52	0.4437\\
0.56	0.3772\\
0.6	0.2954\\
0.64	0.1975\\
0.68	0.0989\\
0.72	0.029\\
0.76	0.0051\\
0.8	0.0004\\
0.84	0\\
0.88	0\\
0.92	0\\
0.96	0\\
1	0\\
};
\addplot[color=black,line width=0.25mm, ->,>=stealth
]
table[row sep=crcr]{%
0.95  0.95 \\
0.4   0.5\\
};
\node[above] at (axis cs:0.25,0.4) {\small{$\mathbf{ a_s=[500 \colon 200 \colon 1500]}$}};
\end{axis}

\begin{axis}[%
width=2.911in,
height=2.051in,
at={(0in,0in)},
scale only axis,
xmin=0,
xmax=1,
ymin=0,
ymax=1,
axis line style={draw=none},
ticks=none,
axis x line*=bottom,
axis y line*=left,
legend style={legend cell align=left, align=left, draw=white!15!black}
]
\end{axis}
\end{tikzpicture}%
\caption{\small{Number of LEOSATs ($N_s =1300$).}}
\label{Fig:P_Cov_Meta_Distribution_Ns_1300}
  \end{subfigure}
  \caption{\small{Meta distribution of conditional uplink coverage probability $(\bar{F}_{\rm Cov}(\gamma_{\rm th}, \tau))$ for distinct LEOSATs altitude $a_s$.}}
\end{figure*}

\begin{figure*}[!t]
  \centering
  \begin{subfigure}[t]{0.5\textwidth}
      \centering
%
%
\definecolor{mycolor1}{rgb}{0.85098,0.32549,0.09804}%
\definecolor{mycolor2}{rgb}{0.92941,0.69412,0.12549}%
\definecolor{mycolor3}{rgb}{0.46667,0.67451,0.18824}%
\begin{tikzpicture}

\begin{axis}[%
width=2.911in,
height=2.051in,
at={(0in,0in)},
scale only axis,
xmin=8,
xmax=20,
xlabel style={font=\color{white!15!black}},
xlabel={t [time slots]},
ymin=0,
ymax=1,
ylabel style={font=\color{white!15!black}},
ylabel={$F_{\bar{A}_k}(\gamma_{\rm th},t)$},
axis background/.style={fill=white},
xmajorgrids,
xminorgrids,
ymajorgrids,
yminorgrids,
legend style={ at={(0.02,0.1)}, anchor=south west, legend columns=1,legend cell align=left, align=left, draw=white!15!black}
]
\addplot [color=black, line width=2.0pt]
  table[row sep=crcr]{%
8	1\\
9	1\\
10	1\\
11	1\\
12	1\\
13	1\\
14	1\\
15	1\\
16	1\\
17	1\\
18	1\\
19	1\\
20	1\\
21	1\\
22	1\\
23	1\\
24	1\\
25	1\\
26	1\\
27	1\\
28	1\\
29	1\\
30	1\\
31	1\\
32	1\\
33	1\\
34	1\\
35	1\\
};
\addlegendentry{Analysis}

\addplot [only marks,color=black, line width=2.0pt, draw=none, mark=triangle, mark options={solid, black}]
  table[row sep=crcr]{%
8	1\\
9	1\\
10	1\\
11	1\\
12	1\\
13	1\\
14	1\\
15	1\\
16	1\\
17	1\\
18	1\\
19	1\\
20	1\\
21	1\\
22	1\\
23	1\\
24	1\\
25	1\\
26	1\\
27	1\\
28	1\\
29	1\\
30	1\\
31	1\\
32	1\\
33	1\\
34	1\\
35	1\\
};
\addlegendentry{Simulation}

\addplot [color=mycolor1, line width=2.0pt, forget plot]
  table[row sep=crcr]{%
8	0.999999993283809\\
9	0.99999999999666\\
10	0.999999999999994\\
11	1\\
12	1\\
13	1\\
14	1\\
15	1\\
16	1\\
17	1\\
18	1\\
19	1\\
20	1\\
21	1\\
22	1\\
23	1\\
24	1\\
25	1\\
26	1\\
27	1\\
28	1\\
29	1\\
30	1\\
31	1\\
32	1\\
33	1\\
34	1\\
35	1\\
};
\addplot [only marks,color=mycolor1, line width=2.0pt, draw=none, mark=triangle, mark options={solid, mycolor1}, forget plot]
  table[row sep=crcr]{%
8	1\\
9	1\\
10	1\\
11	1\\
12	1\\
13	1\\
14	1\\
15	1\\
16	1\\
17	1\\
18	1\\
19	1\\
20	1\\
21	1\\
22	1\\
23	1\\
24	1\\
25	1\\
26	1\\
27	1\\
28	1\\
29	1\\
30	1\\
31	1\\
32	1\\
33	1\\
34	1\\
35	1\\
};
\addplot [color=mycolor2, line width=2.0pt, forget plot]
  table[row sep=crcr]{%
8	0.999920816316702\\
9	0.99999844936487\\
10	0.999999941483151\\
11	0.999999996442697\\
12	0.999999999690548\\
13	0.999999999964449\\
14	0.999999999994902\\
15	0.999999999999124\\
16	0.999999999999825\\
17	0.99999999999996\\
18	0.99999999999999\\
19	0.999999999999997\\
20	0.999999999999999\\
21	1\\
22	1\\
23	1\\
24	1\\
25	1\\
26	1\\
27	1\\
28	1\\
29	1\\
30	1\\
31	1\\
32	1\\
33	1\\
34	1\\
35	1\\
};
\addplot [only marks,color=mycolor2, line width=2.0pt, draw=none, mark=triangle, mark options={solid, mycolor2}, forget plot]
  table[row sep=crcr]{%
8	0.9998\\
9	1\\
10	1\\
11	1\\
12	1\\
13	1\\
14	1\\
15	1\\
16	1\\
17	1\\
18	1\\
19	1\\
20	1\\
21	1\\
22	1\\
23	1\\
24	1\\
25	1\\
26	1\\
27	1\\
28	1\\
29	1\\
30	1\\
31	1\\
32	1\\
33	1\\
34	1\\
35	1\\
};
\addplot [color=blue, line width=2.0pt, forget plot]
  table[row sep=crcr]{%
8	0.99362873845879\\
9	0.999354839973584\\
10	0.999906527329649\\
11	0.999982265934146\\
12	0.999995864475769\\
13	0.999998866171867\\
14	0.999999645986825\\
15	0.99999987704143\\
16	0.999999953326141\\
17	0.99999998090031\\
18	0.999999991664294\\
19	0.99999999615338\\
20	0.999999998136217\\
21	0.999999999057274\\
22	0.999999999504592\\
23	0.999999999730615\\
24	0.99999999984895\\
25	0.999999999912917\\
26	0.999999999948513\\
27	0.999999999968849\\
28	0.999999999980752\\
29	0.999999999987873\\
30	0.999999999992222\\
31	0.999999999994928\\
32	0.999999999996641\\
33	0.999999999997744\\
34	0.999999999998464\\
35	0.999999999998941\\
};
\addplot [only marks,color=blue, line width=2.0pt, draw=none, mark=triangle, mark options={solid, blue}, forget plot]
  table[row sep=crcr]{%
8	0.991\\
9	0.9991\\
10	1\\
11	1\\
12	1\\
13	1\\
14	1\\
15	1\\
16	1\\
17	1\\
18	1\\
19	1\\
20	1\\
21	1\\
22	1\\
23	1\\
24	1\\
25	1\\
26	1\\
27	1\\
28	1\\
29	1\\
30	1\\
31	1\\
32	1\\
33	1\\
34	1\\
35	1\\
};
\addplot [color=red, line width=2.0pt, forget plot]
  table[row sep=crcr]{%
8	0.942761728100728\\
9	0.986115920483599\\
10	0.995870825980991\\
11	0.998557965269178\\
12	0.999428930447445\\
13	0.999750068319517\\
14	0.99988138520221\\
15	0.999939816580498\\
16	0.999967705622268\\
17	0.99998182713398\\
18	0.999989347239357\\
19	0.999993529975584\\
20	0.999995946339703\\
21	0.99999738957896\\
22	0.99999827742159\\
23	0.999998838172009\\
24	0.999999200803779\\
25	0.999999440368408\\
26	0.999999601721659\\
27	0.999999712328854\\
28	0.999999789381495\\
29	0.999999843859084\\
30	0.999999882904452\\
31	0.999999911244114\\
32	0.99999993205519\\
33	0.999999947504612\\
34	0.999999959090447\\
35	0.999999967861481\\
};
\addplot [only marks,color=red, line width=2.0pt, draw=none, mark=triangle, mark options={solid, red}, forget plot]
  table[row sep=crcr]{%
8	0.9182\\
9	0.9843\\
10	0.9982\\
11	0.9996\\
12	1\\
13	1\\
14	1\\
15	1\\
16	1\\
17	1\\
18	1\\
19	1\\
20	1\\
21	1\\
22	1\\
23	1\\
24	1\\
25	1\\
26	1\\
27	1\\
28	1\\
29	1\\
30	1\\
31	1\\
32	1\\
33	1\\
34	1\\
35	1\\
};
\addplot [color=mycolor3, line width=2.0pt, forget plot]
  table[row sep=crcr]{%
8	0.814083581274403\\
9	0.924994147693276\\
10	0.965972675094728\\
11	0.982971974302333\\
12	0.990778655904818\\
13	0.994677082750368\\
14	0.996762424629869\\
15	0.997943300434699\\
16	0.998644712761348\\
17	0.999078567460596\\
18	0.999356423210777\\
19	0.999539817792981\\
20	0.999664098632954\\
21	0.999750300696194\\
22	0.999811338205571\\
23	0.999855362218916\\
24	0.999887646202079\\
25	0.999911678546349\\
26	0.999929813586856\\
27	0.9999436694814\\
28	0.999954377026403\\
29	0.999962738543846\\
30	0.999969331301742\\
31	0.999974576024073\\
32	0.999978783015084\\
33	0.99998218368051\\
34	0.999984952369972\\
35	0.999987221708918\\
};
\addplot [only marks,color=mycolor3, line width=2.0pt, draw=none, mark=triangle, mark options={solid, mycolor3}, forget plot]
  table[row sep=crcr]{%
8	0.7626\\
9	0.8867\\
10	0.9508\\
11	0.981\\
12	0.9939\\
13	0.9979\\
14	0.9993\\
15	0.9998\\
16	1\\
17	1\\
18	1\\
19	1\\
20	1\\
21	1\\
22	1\\
23	1\\
24	1\\
25	1\\
26	1\\
27	1\\
28	1\\
29	1\\
30	1\\
31	1\\
32	1\\
33	1\\
34	1\\
35	1\\
};
\addplot[color=black,line width=0.25mm, ->,>=stealth
]
table[row sep=crcr]{%
8  0.99 \\
14   0.6\\
};
\node[above] at (axis cs:14,0.4) {\small{$\mathbf{ a_s=[500 \colon 200 \colon 1500]}$}};
\end{axis}

\begin{axis}[%
width=2.911in,
height=2.051in,
at={(0in,0in)},
scale only axis,
xmin=0,
xmax=1,
ymin=0,
ymax=1,
axis line style={draw=none},
ticks=none,
axis x line*=bottom,
axis y line*=left,
legend style={legend cell align=left, align=left, draw=white!15!black}
]
\end{axis}
\end{tikzpicture}%
\caption{\small{Number of LEOSATs ($N_s =900$).}}
\label{Fig:PAoI_Distribution_Ns_900}
  \end{subfigure}%
  ~ 
    \begin{subfigure}[t]{0.5\textwidth}
    \centering
%
%
\definecolor{mycolor1}{rgb}{0.49412,0.18431,0.55686}%
\definecolor{mycolor2}{rgb}{0.92941,0.69412,0.12549}%
\definecolor{mycolor3}{rgb}{0.46667,0.67451,0.18824}%
\begin{tikzpicture}

\begin{axis}[%
width=2.911in,
height=2.051in,
at={(0in,0in)},
scale only axis,
xmin=8,
xmax=20,
xlabel style={font=\color{white!15!black}},
xlabel={t [time slots]},
ymin=0,
ymax=1,
ylabel style={font=\color{white!15!black}},
ylabel={$F_{\bar{A}_k}(\gamma_{\rm th},t)$},
axis background/.style={fill=white},
xmajorgrids,
xminorgrids,
ymajorgrids,
yminorgrids,
legend style={ at={(0.02,0.1)}, anchor=south west, legend columns=1,legend cell align=left, align=left, draw=white!15!black}
]
\addplot [color=black, line width=2.0pt]
  table[row sep=crcr]{%
8	0.999999992592411\\
9	0.999999999998558\\
10	0.999999999999999\\
11	1\\
12	1\\
13	1\\
14	1\\
15	1\\
16	1\\
17	1\\
18	1\\
19	1\\
20	1\\
21	1\\
22	1\\
23	1\\
24	1\\
25	1\\
26	1\\
27	1\\
28	1\\
29	1\\
30	1\\
31	1\\
32	1\\
33	1\\
34	1\\
35	1\\
};
\addlegendentry{Analysis}

\addplot [only marks,color=black, line width=2.0pt, draw=none, mark=triangle, mark options={solid, black}]
  table[row sep=crcr]{%
8	1\\
9	1\\
10	1\\
11	1\\
12	1\\
13	1\\
14	1\\
15	1\\
16	1\\
17	1\\
18	1\\
19	1\\
20	1\\
21	1\\
22	1\\
23	1\\
24	1\\
25	1\\
26	1\\
27	1\\
28	1\\
29	1\\
30	1\\
31	1\\
32	1\\
33	1\\
34	1\\
35	1\\
};
\addlegendentry{Simulation}

\addplot [color=mycolor1, line width=2.0pt, forget plot]
  table[row sep=crcr]{%
8	0.999720505912612\\
9	0.999992712961723\\
10	0.999999657253859\\
11	0.9999999750539\\
12	0.999999997472058\\
13	0.999999999668379\\
14	0.99999999994652\\
15	0.99999999998979\\
16	0.999999999997757\\
17	0.999999999999445\\
18	0.999999999999848\\
19	0.999999999999955\\
20	0.999999999999985\\
21	0.999999999999995\\
22	0.999999999999998\\
23	0.999999999999999\\
24	1\\
25	1\\
26	1\\
27	1\\
28	1\\
29	1\\
30	1\\
31	1\\
32	1\\
33	1\\
34	1\\
35	1\\
};
\addplot [only marks,color=mycolor1, line width=2.0pt, draw=none, mark=triangle, mark options={solid, mycolor1}, forget plot]
  table[row sep=crcr]{%
8	0.9998\\
9	1\\
10	1\\
11	1\\
12	1\\
13	1\\
14	1\\
15	1\\
16	1\\
17	1\\
18	1\\
19	1\\
20	1\\
21	1\\
22	1\\
23	1\\
24	1\\
25	1\\
26	1\\
27	1\\
28	1\\
29	1\\
30	1\\
31	1\\
32	1\\
33	1\\
34	1\\
35	1\\
};
\addplot [color=mycolor2, line width=2.0pt, forget plot]
  table[row sep=crcr]{%
8	0.977619863396152\\
9	0.996737936916132\\
10	0.99937323212503\\
11	0.999850172291147\\
12	0.999957500716088\\
13	0.999986187624609\\
14	0.999994987607659\\
15	0.999998007867957\\
16	0.999999145569557\\
17	0.999999609042163\\
18	0.999999810885931\\
19	0.999999903998789\\
20	0.999999949161553\\
21	0.999999972053619\\
22	0.999999984118746\\
23	0.99999999070267\\
24	0.999999994409391\\
25	0.999999996555842\\
26	0.99999999783092\\
27	0.999999998606182\\
28	0.999999999087679\\
29	0.999999999392622\\
30	0.999999999589253\\
31	0.999999999718168\\
32	0.999999999804\\
33	0.999999999861971\\
34	0.999999999901652\\
35	0.999999999929154\\
};
\addplot [only marks,color=mycolor2, line width=2.0pt, draw=none, mark=triangle, mark options={solid, mycolor2}, forget plot]
  table[row sep=crcr]{%
8	0.9703\\
9	0.9977\\
10	1\\
11	1\\
12	1\\
13	1\\
14	1\\
15	1\\
16	1\\
17	1\\
18	1\\
19	1\\
20	1\\
21	1\\
22	1\\
23	1\\
24	1\\
25	1\\
26	1\\
27	1\\
28	1\\
29	1\\
30	1\\
31	1\\
32	1\\
33	1\\
34	1\\
35	1\\
};
\addplot [color=blue, line width=2.0pt, forget plot]
  table[row sep=crcr]{%
8	0.848972680784577\\
9	0.950170112623716\\
10	0.981227410802041\\
11	0.992055326503372\\
12	0.996300931661357\\
13	0.998137977538247\\
14	0.999000432106212\\
15	0.999433809927929\\
16	0.999664395939274\\
17	0.999793199948333\\
18	0.999868222686847\\
19	0.999913536139176\\
20	0.999941789587968\\
21	0.99995990732773\\
22	0.999971818908636\\
23	0.999979826911382\\
24	0.999985319733238\\
25	0.999989156354312\\
26	0.999991880708893\\
27	0.99999384457991\\
28	0.999995279899561\\
29	0.999996342300216\\
30	0.999997137920005\\
31	0.999997740231925\\
32	0.999998200804245\\
33	0.999998556297777\\
34	0.999998833087856\\
35	0.999999050362101\\
};
\addplot [only marks,color=blue, line width=2.0pt, draw=none, mark=triangle, mark options={solid, blue}, forget plot]
  table[row sep=crcr]{%
8	0.8112\\
9	0.9268\\
10	0.9779\\
11	0.9943\\
12	0.9982\\
13	0.9994\\
14	0.9999\\
15	0.9999\\
16	0.9999\\
17	0.9999\\
18	1\\
19	1\\
20	1\\
21	1\\
22	1\\
23	1\\
24	1\\
25	1\\
26	1\\
27	1\\
28	1\\
29	1\\
30	1\\
31	1\\
32	1\\
33	1\\
34	1\\
35	1\\
};
\addplot [color=red, line width=2.0pt, forget plot]
  table[row sep=crcr]{%
8	0.625388683502721\\
9	0.807672290869055\\
10	0.895186225207176\\
11	0.939215235404373\\
12	0.962769128813557\\
13	0.976114253858151\\
14	0.984064389023567\\
15	0.989009412977037\\
16	0.992201492358629\\
17	0.994329125361536\\
18	0.99578734542818\\
19	0.996811462634271\\
20	0.997546356389739\\
21	0.99808388019864\\
22	0.998483808224613\\
23	0.998785957595657\\
24	0.999017414319827\\
25	0.999196956863543\\
26	0.999337830302173\\
27	0.999449524620549\\
28	0.999538937666898\\
29	0.999611149735351\\
30	0.999669948220272\\
31	0.999718188532323\\
32	0.999758045957314\\
33	0.999791193759461\\
34	0.999818930691062\\
35	0.999842273335773\\
};
\addplot [only marks,color=red, line width=2.0pt, draw=none, mark=triangle, mark options={solid, red}, forget plot]
  table[row sep=crcr]{%
8	0.637\\
9	0.7746\\
10	0.8538\\
11	0.9093\\
12	0.9484\\
13	0.9717\\
14	0.9841\\
15	0.9922\\
16	0.996\\
17	0.9982\\
18	0.999\\
19	0.9995\\
20	0.9997\\
21	0.9998\\
22	1\\
23	1\\
24	1\\
25	1\\
26	1\\
27	1\\
28	1\\
29	1\\
30	1\\
31	1\\
32	1\\
33	1\\
34	1\\
35	1\\
};
\addplot [color=mycolor3, line width=2.0pt, forget plot]
  table[row sep=crcr]{%
8	0.405988054391799\\
9	0.605543767578081\\
10	0.73237497592992\\
11	0.812571328313163\\
12	0.864601691377783\\
13	0.899437489711912\\
14	0.923481969638564\\
15	0.940542969034113\\
16	0.952949756794903\\
17	0.962169952409726\\
18	0.969154761643755\\
19	0.97453696024985\\
20	0.978747591003244\\
21	0.982086648936229\\
22	0.984767019023065\\
23	0.986942453659052\\
24	0.988725787901452\\
25	0.990201043760011\\
26	0.991431626180081\\
27	0.992465969329093\\
28	0.993341488723885\\
29	0.994087388411391\\
30	0.994726682127484\\
31	0.99527766692802\\
32	0.99575501022999\\
33	0.996170560438255\\
34	0.996533957600874\\
35	0.996853097800407\\
};
\addplot [only marks,color=mycolor3, line width=2.0pt, draw=none, mark=triangle, mark options={solid, mycolor3}, forget plot]
  table[row sep=crcr]{%
8	0.4716\\
9	0.6062\\
10	0.6922\\
11	0.7546\\
12	0.803\\
13	0.8416\\
14	0.8737\\
15	0.9016\\
16	0.924\\
17	0.9413\\
18	0.9555\\
19	0.9652\\
20	0.9746\\
21	0.9799\\
22	0.9849\\
23	0.989\\
24	0.9911\\
25	0.9932\\
26	0.9947\\
27	0.9958\\
28	0.9969\\
29	0.9977\\
30	0.9985\\
31	0.9989\\
32	0.9991\\
33	0.9993\\
34	0.9997\\
35	0.9997\\
};
\addplot[color=black,line width=0.25mm, ->,>=stealth
]
table[row sep=crcr]{%
8  0.99 \\
14   0.6\\
};
\node[above] at (axis cs:14,0.4) {\small{$\mathbf{ a_s=[500 \colon 200 \colon 1500]}$}};
\end{axis}

\begin{axis}[%
width=2.911in,
height=2.051in,
at={(0in,0in)},
scale only axis,
xmin=0,
xmax=1,
ymin=0,
ymax=1,
axis line style={draw=none},
ticks=none,
axis x line*=bottom,
axis y line*=left,
legend style={legend cell align=left, align=left, draw=white!15!black}
]
\end{axis}
\end{tikzpicture}%
\caption{\small{Number of LEOSATs ($N_s =1300$).}}
\label{Fig:PAoI_Distribution_Ns_1300}
  \end{subfigure}
  \caption{\small{Spatial distribution of the conditional mean PAoI ($F_{\bar{A}_k}(\gamma_{\rm th},t)$) for distinct LEOSATs altitude $a_s$.}}
\end{figure*}
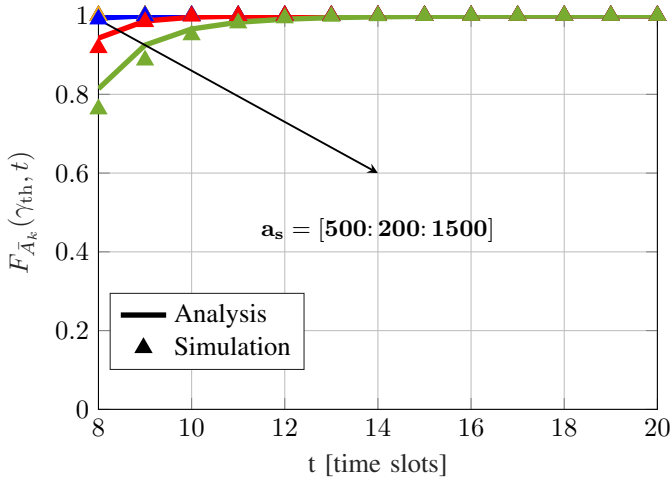
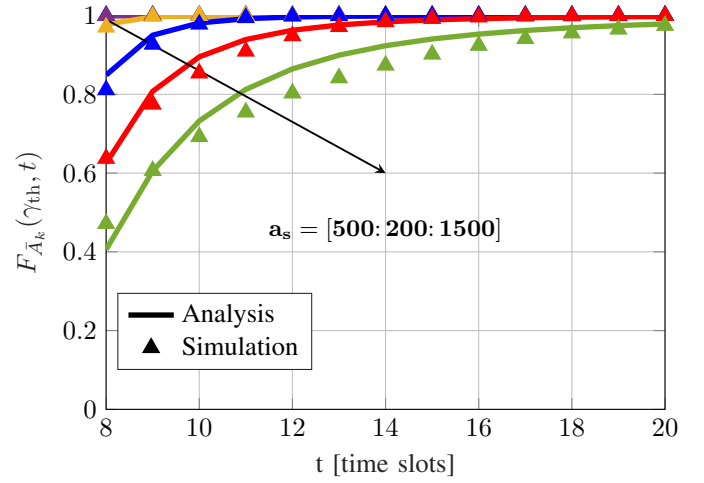
\begin{theorem}[Spatial Distribution of Conditional Mean PAoI]
Given $\psi_n$ and SINR threshold $\gamma_{\rm th}$ and for a given time deadline $t$. The distribution of the conditional mean PAoI is approximated as:
\label{thm:Spatial_Distribution_of_Conditional_Mean_PAoI}
\begin{align}\label{Eq:Spatial_Distribution_of_Conditional_Mean_PAoI}
F_{\bar{A}_k}(\gamma_{\rm th},t) = 1 - \mathcal{I}_\omega \left(\alpha,\beta\right),
\end{align}
where $a$ and $b$ are defined in \eqref{Eq:Beta}, and \\
$\omega = 2\left(t-\dfrac{2}{(1-P_{v})} + \dfrac{1}{\lambda_a}+ \dfrac{1}{c}\int_{a_s}^{ {\rm r_{\rm max}}} r_o f_{R_o}(r_o) dr_o\right)^{-1}$.
\end{theorem}
\begin{IEEEproof}
See Appendix~\ref{app:Spatial_Distribution_of_Conditional_Mean_PAoI}.
\end{IEEEproof} 

Having established the theoretical derivations in the preceding section, we now turn to the numerical results to reveal useful insights.

\section{Numerical Results} \label{Sec:Numerical_Results}
\begin{figure*}[!t]
  \centering
  \begin{subfigure}[t]{0.5\textwidth}
      \centering
%
%
\definecolor{mycolor1}{rgb}{0.87059,0.49020,0.00000}%
\definecolor{mycolor2}{rgb}{0.00000,0.49804,0.00000}%
\definecolor{mycolor3}{rgb}{0.60000,0.20000,0.00000}%
\definecolor{mycolor4}{rgb}{0.00000,0.74902,0.74902}%
\definecolor{mycolor5}{rgb}{0.74902,0.00000,0.74902}%
\begin{tikzpicture}

\begin{axis}[%
width=2.911in,
height=2.051in,
at={(0in,0in)},
scale only axis,
xmin=500,
xmax=1500,
xlabel style={font=\color{white!15!black}},
xlabel={$a_s$ [km]},
ymin=6,
ymax=18,
ylabel style={font=\color{white!15!black}},
ylabel={$M^1(\bar{A_k})$ [time slots]},
axis background/.style={fill=white},
xmajorgrids,
xminorgrids,
ymajorgrids,
yminorgrids,
legend style={at={(0.02,0.35)},legend columns=2,column sep=5pt, anchor=south west, legend cell align=left, align=left, draw=white!15!black,font=\footnotesize,legend image post style={scale=0.5}} 
]
\addplot [color=black, line width=2.0pt]
  table[row sep=crcr]{%
500	7.50188166730092\\
700	6.44518548124496\\
900	6.15207527129142\\
1100	6.11726556947773\\
1300	6.17732746341281\\
1500	6.28066661406484\\
};
\addlegendentry{Analysis, $N_s=500$}

\addplot [only marks,color=black, line width=2.0pt, draw=none, mark=triangle, mark options={solid, black}]
  table[row sep=crcr]{%
500	7.50292615557025\\
700	6.44766222580918\\
900	6.15654930418309\\
1100	6.12318392180485\\
1300	6.18352799155919\\
1500	6.29089804791968\\
};
\addlegendentry{Simu, $N_s=500$}

\addplot [color=mycolor1, line width=2.0pt, forget plot]
  table[row sep=crcr]{%
500	6.89346930306762\\
700	6.25073733868191\\
900	6.16180617142358\\
1100	6.24599413427484\\
1300	6.39876928582156\\
1500	6.60675953731001\\
};
\addlegendimage{color=mycolor1, line width=2.0pt, forget plot}
\addlegendentry{Analysis, $N_s=700$}

\addplot [only marks,color=mycolor1, line width=2.0pt, draw=none, mark=*, mark options={solid, mycolor1}, forget plot]
  table[row sep=crcr]{%
500	6.89759073264184\\
700	6.25498069733527\\
900	6.16637165079461\\
1100	6.25117418820336\\
1300	6.40650918080838\\
1500	6.62497829882547\\
};
\addlegendimage{only marks,color=mycolor1, line width=2.0pt, draw=none, mark=*, mark options={solid, mycolor1}, forget plot}
\addlegendentry{Simu, $N_s=700$}

\addplot [color=mycolor2, line width=2.0pt, forget plot]
  table[row sep=crcr]{%
500	6.59536326575063\\
700	6.21299203034182\\
900	6.26178041893328\\
1100	6.44413189851513\\
1300	6.7042781074439\\
1500	7.05204272444423\\
};
\addlegendimage{color=mycolor2, line width=2.0pt, forget plot}
\addlegendentry{Analysis, $N_s=900$}

\addplot [only marks,color=mycolor2, line width=2.0pt, draw=none, mark=triangle, mark options={solid, rotate=90, mycolor2}, forget plot]
  table[row sep=crcr]{%
500	6.60074816706974\\
700	6.21601149053826\\
900	6.26669196903357\\
1100	6.44660377303544\\
1300	6.70386713906044\\
1500	7.06873115178265\\
};
\addlegendimage{only marks,color=mycolor2, line width=2.0pt, draw=none, mark=triangle, mark options={solid, rotate=90, mycolor2}, forget plot}
\addlegendentry{Simu, $N_s=900$}

\addplot [color=mycolor3, line width=2.0pt, forget plot]
  table[row sep=crcr]{%
500	6.44226127904128\\
700	6.25081153363729\\
900	6.41269523532657\\
1100	6.70058966752926\\
1300	7.09643048841175\\
1500	7.62841151614451\\
};
\addlegendimage{color=mycolor3, line width=2.0pt, forget plot}
\addlegendentry{Analysis, $N_s=1100$}

\addplot [only marks,color=mycolor3, line width=2.0pt, draw=none, mark=square, mark options={solid, mycolor3}, forget plot]
  table[row sep=crcr]{%
500	6.44848752207689\\
700	6.25264740931968\\
900	6.40707211217255\\
1100	6.69945630226864\\
1300	7.0809805855558\\
1500	7.64921419835223\\
};
\addlegendimage{only marks,color=mycolor3, line width=2.0pt, draw=none, mark=square, mark options={solid, mycolor3}, forget plot}
\addlegendentry{Simu, $N_s=1100$}

\addplot [color=blue, line width=2.0pt, forget plot]
  table[row sep=crcr]{%
500	6.37041592038898\\
700	6.33393280485361\\
900	6.60404884536709\\
1100	7.01653841097008\\
1300	7.58384632323186\\
1500	8.35574610102325\\
};
\addlegendimage{color=blue, line width=2.0pt, forget plot}
\addlegendentry{Analysis, $N_s=1300$}

\addplot [only marks,color=blue, line width=2.0pt, draw=none, mark size=4.0pt, mark=x, mark options={solid, blue}, forget plot]
  table[row sep=crcr]{%
500	6.37430184033973\\
700	6.33586828027564\\
900	6.59824800603143\\
1100	7.00355228827415\\
1300	7.54864581798437\\
1500	8.36800845603595\\
};
\addlegendimage{only marks,color=blue, line width=1.5pt, draw=none, mark size=4.0pt, mark=x, mark options={solid, blue}, forget plot}
\addlegendentry{Simu, $N_s=1300$}

\addplot [color=red, line width=2.0pt, forget plot]
  table[row sep=crcr]{%
500	6.34999402258377\\
700	6.44942647731167\\
900	6.83396208907626\\
1100	7.39671723609233\\
1300	8.1791976958579\\
1500	9.26136971814269\\
};
\addlegendimage{color=red, line width=2.0pt, forget plot}
\addlegendentry{Analysis, $N_s=1500$}

\addplot [only marks,color=red, line width=2.0pt, draw=none, mark=triangle, mark options={solid, rotate=180, red}, forget plot]
  table[row sep=crcr]{%
500	6.35659414892818\\
700	6.44230624570032\\
900	6.80704001689557\\
1100	7.36537315669896\\
1300	8.14802573155321\\
1500	9.30146770672553\\
};
\addlegendimage{only marks,color=red, line width=2.0pt, draw=none, mark=triangle, mark options={solid, rotate=180, red}, forget plot}
\addlegendentry{Simu, $N_s=1500$}

\addplot [color=mycolor4, line width=2.0pt, forget plot]
  table[row sep=crcr]{%
500	6.36520953849155\\
700	6.59160974839532\\
900	7.1036984867859\\
1100	7.8479720072822\\
1300	8.89900341257902\\
1500	10.3803926016584\\
};
\addlegendimage{color=mycolor4, line width=2.0pt, forget plot}
\addlegendentry{Analysis, $N_s=1700$}

\addplot [only marks,color=mycolor4, line width=2.0pt, draw=none, mark=triangle, mark options={solid, rotate=270, mycolor4}, forget plot]
  table[row sep=crcr]{%
500	6.36565679371472\\
700	6.56486480231359\\
900	7.06190102543079\\
1100	7.7907329034475\\
1300	8.87948400058774\\
1500	10.3872514751117\\
};
\addlegendimage{only marks,color=mycolor4, line width=2.0pt, draw=none, mark=triangle, mark options={solid, rotate=270, mycolor4}, forget plot}
\addlegendentry{Simu, $N_s=1700$}

\addplot [color=mycolor5, line width=2.0pt, forget plot]
  table[row sep=crcr]{%
500	6.4070062514208\\
700	6.75823117563892\\
900	7.41601260313497\\
1100	8.37900190363852\\
1300	9.7637918967143\\
1500	11.756635579803\\
};
\addlegendimage{color=mycolor5, line width=2.0pt, forget plot}
\addlegendentry{Analysis, $N_s=1900$}

\addplot [only marks,color=mycolor5, line width=2.0pt, draw=none, mark=diamond, mark options={solid, mycolor5}, forget plot]
  table[row sep=crcr]{%
500	6.40299715449584\\
700	6.73083863507546\\
900	7.35526682030321\\
1100	8.31175663329155\\
1300	9.75183579418865\\
1500	11.8691137363744\\
};
\addlegendimage{only marks,color=mycolor5, line width=2.0pt, draw=none, mark=diamond, mark options={solid, mycolor5}, forget plot}
\addlegendentry{Simu, $N_s=1900$}
\end{axis}

\begin{axis}[%
width=2.911in,
height=2.051in,
at={(0in,0in)},
scale only axis,
xmin=0,
xmax=1,
ymin=0,
ymax=1,
axis line style={draw=none},
ticks=none,
axis x line*=bottom,
axis y line*=left,
legend style={legend cell align=left, align=left, draw=white!15!black}
]
\end{axis}
\end{tikzpicture}%
\caption{\small{Spatial average of the Conditional PAoI with a distinct number of LEOSATs $N_s$.}}
\label{Fig:Avg_PAoI_vs_h}
  \end{subfigure}%
  ~ 
    \begin{subfigure}[t]{0.5\textwidth}
    \centering
%
%
\definecolor{mycolor1}{rgb}{0.87059,0.49020,0.00000}%
\definecolor{mycolor2}{rgb}{0.00000,0.49804,0.00000}%
\definecolor{mycolor3}{rgb}{0.49412,0.18431,0.55686}%
\begin{tikzpicture}

\begin{axis}[%
width=2.911in,
height=2.051in,
at={(0in,0in)},
scale only axis,
xmin=500,
xmax=2000,
xlabel style={font=\color{white!15!black}},
xlabel={Ns},
ymin=6,
ymax=16,
ylabel style={font=\color{white!15!black}},
ylabel={Mean PAoI [time slots]},
axis background/.style={fill=white},
xmajorgrids,
xminorgrids,
ymajorgrids,
yminorgrids,
legend style={at={(0.02,0.52)},legend columns=2,column sep=5pt, anchor=south west, legend cell align=left, align=left, draw=white!15!black,font=\footnotesize,legend image post style={scale=0.5}} 
]
\addplot [color=black, line width=2.0pt]
  table[row sep=crcr]{%
500	7.50188166730092\\
700	6.89346930306762\\
900	6.59536326575063\\
1100	6.44226127904128\\
1300	6.37041592038898\\
1500	6.34999402258377\\
1700	6.36520953849155\\
1900	6.4070062514208\\
};
\addlegendentry{Analysis, $a_s=500$}

\addplot [only marks,color=black, line width=2.0pt, draw=none,mark size=3.0pt,  mark=triangle, mark options={solid, black}]
  table[row sep=crcr]{%
500	7.50292615557025\\
700	6.89759073264184\\
900	6.60074816706974\\
1100	6.44848752207689\\
1300	6.37430184033973\\
1500	6.35659414892818\\
1700	6.36565679371472\\
1900	6.40299715449584\\
};
\addlegendentry{Simu, $a_s=500$}

\addplot [color=mycolor1, line width=2.0pt, forget plot]
  table[row sep=crcr]{%
500	6.44518548124496\\
700	6.25073733868191\\
900	6.21299203034182\\
1100	6.25081153363729\\
1300	6.33393280485361\\
1500	6.44942647731167\\
1700	6.59160974839532\\
1900	6.75823117563892\\
};
\addlegendimage{color=mycolor1, line width=2.0pt, forget plot}
\addlegendentry{Analysis, $a_s=700$}
\addplot [only marks,color=mycolor1, line width=2.0pt, draw=none,mark size=2.0pt,  mark=*, mark options={solid, mycolor1}, forget plot]
  table[row sep=crcr]{%
500	6.44766222580918\\
700	6.25498069733527\\
900	6.21601149053826\\
1100	6.25264740931968\\
1300	6.33586828027564\\
1500	6.44230624570032\\
1700	6.56486480231359\\
1900	6.73083863507546\\
};
\addlegendimage{only marks,color=mycolor1, line width=2.0pt, draw=none,mark size=3.0pt, mark=*, mark options={solid, mycolor1}, forget plot}
\addlegendentry{Simu, $a_s=700$}
\addplot [color=mycolor2, line width=2.0pt, forget plot]
  table[row sep=crcr]{%
500	6.15207527129142\\
700	6.16180617142358\\
900	6.26178041893328\\
1100	6.41269523532657\\
1300	6.60404884536709\\
1500	6.83396208907626\\
1700	7.1036984867859\\
1900	7.41601260313497\\
};
\addlegendimage{color=mycolor2, line width=2.0pt, forget plot}
\addlegendentry{Analysis, $a_s=900$}
\addplot [only marks,color=mycolor2, line width=2.0pt, draw=none,mark size=3.0pt, mark=triangle, mark options={solid, rotate=90, mycolor2}, forget plot]
  table[row sep=crcr]{%
500	6.15654930418309\\
700	6.16637165079461\\
900	6.26669196903357\\
1100	6.40707211217255\\
1300	6.59824800603143\\
1500	6.80704001689557\\
1700	7.06190102543079\\
1900	7.35526682030321\\
};
\addlegendimage{only marks,color=mycolor2, line width=2.0pt, draw=none,mark size=4.0pt, mark=triangle, mark options={solid, rotate=90, mycolor2}, forget plot}
\addlegendentry{Simu, $a_s=900$}
\addplot [color=red, line width=2.0pt, forget plot]
  table[row sep=crcr]{%
500	6.11726556947773\\
700	6.24599413427484\\
900	6.44413189851513\\
1100	6.70058966752926\\
1300	7.01653841097008\\
1500	7.39671723609233\\
1700	7.8479720072822\\
1900	8.37900190363852\\
};
\addlegendimage{color=red, line width=2.0pt, forget plot}
\addlegendentry{Analysis, $a_s=1100$}
\addplot [only marks,color=red, line width=2.0pt, draw=none,mark size=2.0pt,  mark=square, mark options={solid, red}, forget plot]
  table[row sep=crcr]{%
500	6.12318392180485\\
700	6.25117418820336\\
900	6.44660377303544\\
1100	6.69945630226864\\
1300	7.00355228827415\\
1500	7.36537315669896\\
1700	7.7907329034475\\
1900	8.31175663329155\\
};
\addlegendimage{only marks,color=red, line width=1.5pt, draw=none, mark size=4.0pt, mark=square, mark options={solid, red}, forget plot}
\addlegendentry{Simu, $a_s=1100$}
\addplot [color=blue, line width=2.0pt, forget plot]
  table[row sep=crcr]{%
500	6.17732746341281\\
700	6.39876928582156\\
900	6.7042781074439\\
1100	7.09643048841175\\
1300	7.58384632323186\\
1500	8.1791976958579\\
1700	8.89900341257902\\
1900	9.7637918967143\\
};
\addlegendimage{color=blue, line width=2.0pt, forget plot}
\addlegendentry{Analysis, $a_s=1300$}
\addplot [only marks,color=blue, line width=2.0pt, draw=none,mark size=4.0pt,  mark=x, mark options={solid, blue}, forget plot]
  table[row sep=crcr]{%
500	6.18352799155919\\
700	6.40650918080838\\
900	6.70386713906044\\
1100	7.0809805855558\\
1300	7.54864581798437\\
1500	8.14802573155321\\
1700	8.87948400058774\\
1900	9.75183579418865\\
};
\addlegendimage{only marks,color=blue, line width=1.5pt, draw=none, mark size=5.0pt, mark=x, mark options={solid, blue}, forget plot}
\addlegendentry{Simu, $a_s=1300$}
\addplot [color=mycolor3, line width=2.0pt, forget plot]
  table[row sep=crcr]{%
500	6.28066661406484\\
700	6.60675953731001\\
900	7.05204272444423\\
1100	7.62841151614451\\
1300	8.35574610102325\\
1500	9.26136971814269\\
1700	10.3803926016584\\
1900	11.756635579803\\
};
\addlegendimage{color=mycolor3, line width=2.0pt, forget plot}
\addlegendentry{Analysis, $a_s=1500$}
\addplot [only marks,color=mycolor3, line width=2.0pt, draw=none,mark size=3.0pt,  mark=diamond, mark options={solid, mycolor3}, forget plot]
  table[row sep=crcr]{%
500	6.29089804791968\\
700	6.62497829882547\\
900	7.06873115178265\\
1100	7.64921419835223\\
1300	8.36800845603595\\
1500	9.30146770672553\\
1700	10.3872514751117\\
1900	11.8691137363744\\
};
\addlegendimage{only marks,color=mycolor3, line width=2.0pt, draw=none,mark size=4.0pt,  mark=diamond, mark options={solid, mycolor3}, forget plot}
\addlegendentry{Simu, $a_s=1500$}
\end{axis}

\begin{axis}[%
width=2.911in,
height=2.051in,
at={(0in,0in)},
scale only axis,
xmin=0,
xmax=1,
ymin=0,
ymax=1,
axis line style={draw=none},
ticks=none,
axis x line*=bottom,
axis y line*=left,
legend style={legend cell align=left, align=left, draw=white!15!black}
]
\end{axis}
\end{tikzpicture}%
\caption{\small{Spatial average of the Conditional PAoI with a distinct LEOSAT altitude $a_s$.}}
\label{Fig:Avg_PAoI_vs_Ns}
  \end{subfigure}
  \caption{\small{First moment of the conditional mean PAoI $M^1(\bar{A}_k)$ for different number of LEOSATs $N_s$ at distinct LEOSAT altitude $a_s$.}}
\end{figure*}
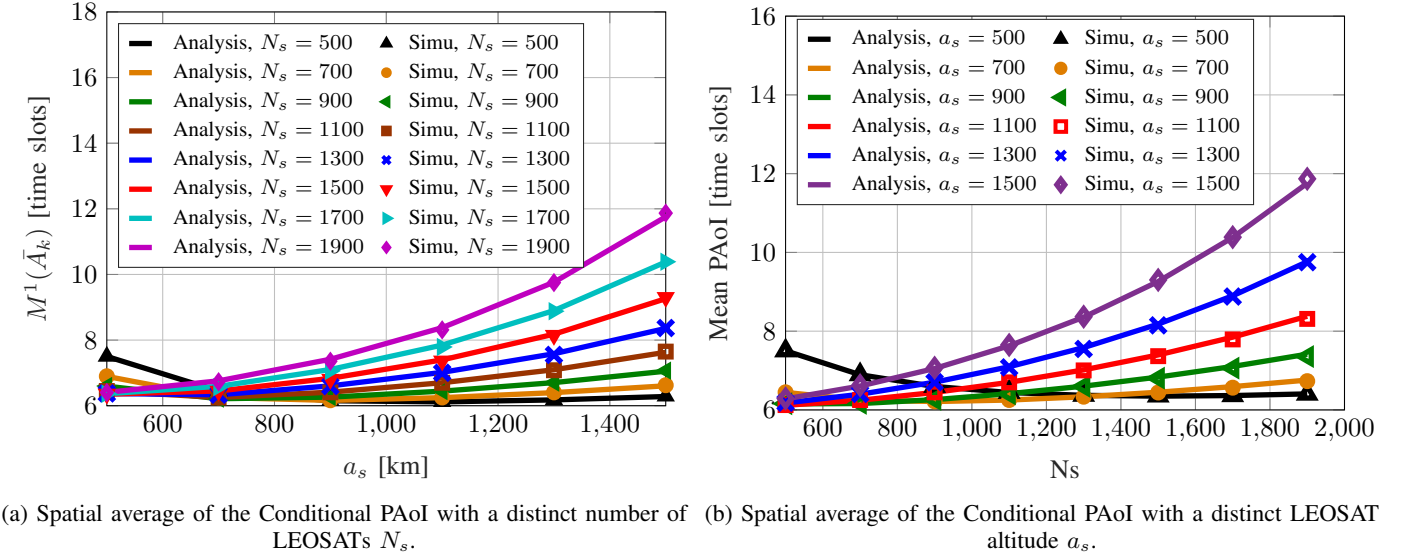

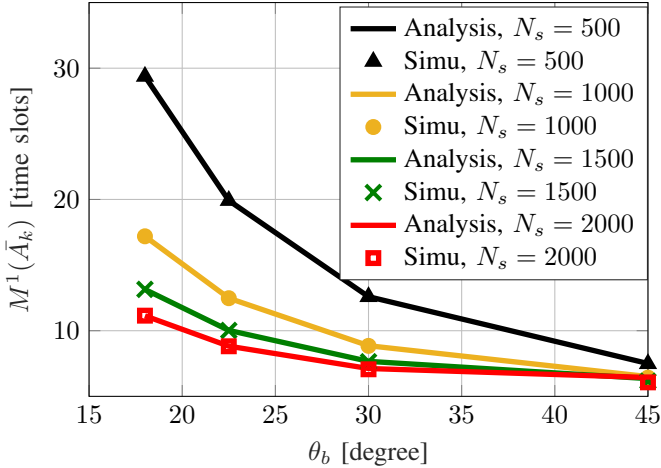
\begin{figure}[t]
\centering
%
%
\definecolor{mycolor1}{rgb}{0.92941,0.69412,0.12549}%
\definecolor{mycolor2}{rgb}{0.00000,0.49804,0.00000}%
\begin{tikzpicture}

\begin{axis}[%
width=2.911in,
height=2.051in,
at={(0in,0in)},
scale only axis,
xmin=15,
xmax=45,
xlabel style={font=\color{white!15!black}},
xlabel={$\theta_b$ [degree]},
ymin=5,
ymax=35,
ylabel style={font=\color{white!15!black}},
ylabel={$M^1(\bar{A_k})$ [time slots]},
axis background/.style={fill=white},
xmajorgrids,
xminorgrids,
ymajorgrids,
yminorgrids,
legend style={at={(1.0,0.99)}, legend cell align=left, align=left, draw=white!15!black}
]
\addplot [color=black, line width=2.0pt]
  table[row sep=crcr]{%
18	29.3744074802142\\
22.5	19.9179783082832\\
30	12.5931351755522\\
45	7.50188166730092\\
};
\addlegendentry{Analysis, $N_s=500$}

\addplot [only marks,color=black, line width=2.0pt, draw=none, mark=triangle, mark options={solid, black}]
  table[row sep=crcr]{%
18	29.3720761527429\\
22.5	19.9150427996349\\
30	12.5883896766342\\
45	7.4839237453424\\
};
\addlegendentry{Simu, $N_s=500$}

\addplot [color=mycolor1, line width=2.0pt, forget plot]
  table[row sep=crcr]{%
18	17.1922046898082\\
22.5	12.4798641722853\\
30	8.85841320566414\\
45	6.50586180535371\\
};
\addlegendimage{color=mycolor1, line width=2.0pt, forget plot}
\addlegendentry{Analysis, $N_s=1000$}
\addplot [only marks,color=mycolor1, line width=2.0pt, draw=none, mark=*, mark options={solid, mycolor1}, forget plot]
  table[row sep=crcr]{%
18	17.1927617862278\\
22.5	12.4771978603708\\
30	8.84598639194824\\
45	6.42635701520983\\
};
\addlegendimage{only marks,color=mycolor1, line width=2.0pt, draw=none, mark=*, mark options={solid, mycolor1}, forget plot}
\addlegendentry{Simu, $N_s=1000$}
\addplot [color=mycolor2, line width=2.0pt, forget plot]
  table[row sep=crcr]{%
18	13.1460345725778\\
22.5	10.0254485980396\\
30	7.66591143040771\\
45	6.34999402258377\\
};
\addlegendimage{color=mycolor2, line width=2.0pt, forget plot}
\addlegendentry{Analysis, $N_s=1500$}
\addplot [only marks,color=mycolor2, line width=1.5pt, draw=none, mark size=4.0pt, mark=x, mark options={solid, mycolor2}, forget plot]
  table[row sep=crcr]{%
18	13.1501321178238\\
22.5	10.0207466706807\\
30	7.63504710108574\\
45	6.15243688537769\\
};
\addlegendimage{only marks,color=mycolor2, line width=1.5pt, draw=none, mark size=4.0pt, mark=x, mark options={solid, mycolor2}, forget plot}
\addlegendentry{Simu, $N_s=1500$}
\addplot [color=red, line width=2.0pt, forget plot]
  table[row sep=crcr]{%
18	11.1371988551447\\
22.5	8.82286673029503\\
30	7.12131976179286\\
45	6.43607072465436\\
};
\addlegendimage{color=red, line width=2.0pt, forget plot}
\addlegendentry{Analysis, $N_s=2000$}
\addplot [only marks,color=red, line width=2.0pt, draw=none, mark=square, mark options={solid, red}, forget plot]
  table[row sep=crcr]{%
18	11.1471885798227\\
22.5	8.81508306160272\\
30	7.06179002220636\\
45	6.06560287403132\\
};
\addlegendimage{only marks,color=red, line width=2.0pt, draw=none, mark=square, mark options={solid, red}, forget plot}
\addlegendentry{Simu, $N_s=2000$}
\end{axis}

\begin{axis}[%
width=2.911in,
height=2.051in,
at={(0in,0in)},
scale only axis,
xmin=0,
xmax=1,
ymin=0,
ymax=1,
axis line style={draw=none},
ticks=none,
axis x line*=bottom,
axis y line*=left,
legend style={legend cell align=left, align=left, draw=white!15!black}
]
\end{axis}
\end{tikzpicture}%
\caption{\small{First moment of the conditional mean PAoI $M^1(\bar{A}_k)$ for different number of LEOSATs $N_s$ with a LEOSAT altitude $a_s =500$ km.}}
\label{Fig:AoI_vs_Different_Theta_b}
\vspace{0.2in}
\end{figure}

\begin{figure}[t]
\centering
%
%
\definecolor{mycolor1}{rgb}{0.92941,0.69412,0.12549}%
\definecolor{mycolor2}{rgb}{0.00000,0.49804,0.00000}%
\begin{tikzpicture}

\begin{axis}[%
width=2.911in,
height=2.051in,
at={(0in,0in)},
scale only axis,
xmin=15,
xmax=45,
xlabel style={font=\color{white!15!black}},
xlabel={$\theta_b$ [degree]},
ymin=6,
ymax=13,
ylabel style={font=\color{white!15!black}},
ylabel={$M^1(\bar{A_k})$ [time slots]},
axis background/.style={fill=white},
xmajorgrids,
xminorgrids,
ymajorgrids,
yminorgrids,
legend style={at={(0.02,0.30)}, anchor=south west, legend columns=1,column sep=2pt,legend cell align=left, align=left, draw=white!15!black}
]
\addplot [color=black, line width=2.0pt]
  table[row sep=crcr]{%
18	7.7651457393046\\
22.5	6.80042835890947\\
30	6.1930325022156\\
45	6.28066661406484\\
};
\addlegendentry{Analysis, $N_s=500$}

\addplot [only marks,color=black, line width=2.0pt, draw=none, mark=triangle, mark options={solid, black}]
  table[row sep=crcr]{%
18	7.76495453624441\\
22.5	6.80417711017998\\
30	6.21113883440796\\
45	6.31111955191801\\
};
\addlegendentry{Simu, $N_s=500$}

\addplot [color=mycolor1, line width=2.0pt, forget plot]
  table[row sep=crcr]{%
18	6.59963605641675\\
22.5	6.25872588296305\\
30	6.23812449465216\\
45	7.32277409204004\\
};
\addlegendimage{color=mycolor1, line width=2.0pt, forget plot}
\addlegendentry{Analysis, $N_s=1000$}
\addplot [only marks,color=mycolor1, line width=2.0pt, draw=none, mark=*, mark options={solid, mycolor1}, forget plot]
  table[row sep=crcr]{%
18	6.62227949742286\\
22.5	6.30130396426188\\
30	6.29778644776598\\
45	7.38881010774589\\
};
\addlegendimage{only marks,color=mycolor1, line width=2.0pt, draw=none, mark=*, mark options={solid, mycolor1}, forget plot}
\addlegendentry{Simu, $N_s=1000$}
\addplot [color=mycolor2, line width=2.0pt, forget plot]
  table[row sep=crcr]{%
18	6.35764489526113\\
22.5	6.29775693477614\\
30	6.61541008620381\\
45	9.26136971814269\\
};
\addlegendimage{color=mycolor2, line width=2.0pt, forget plot}
\addlegendentry{Analysis, $N_s=1500$}
\addplot [only marks,color=mycolor2, line width=1.5pt, draw=none, mark size=4.0pt, mark=x, mark options={solid, mycolor2}, forget plot]
  table[row sep=crcr]{%
18	6.41412922193906\\
22.5	6.38356045372856\\
30	6.73936844491812\\
45	9.27315574283178\\
};
\addlegendimage{only marks,color=mycolor2, line width=1.5pt, draw=none, mark size=4.0pt, mark=x, mark options={solid, mycolor2}, forget plot}
\addlegendentry{Simu, $N_s=1500$}
\addplot [color=red, line width=2.0pt, forget plot]
  table[row sep=crcr]{%
18	6.37477160900633\\
22.5	6.52394308736676\\
30	7.21421865713008\\
45	12.5576250743895\\
};
\addlegendimage{color=red, line width=2.0pt, forget plot}
\addlegendentry{Analysis, $N_s=2000$}
\addplot [only marks,color=red, line width=2.0pt, draw=none, mark=square, mark options={solid, red}, forget plot]
  table[row sep=crcr]{%
18	6.484648274824\\
22.5	6.67864861226177\\
30	7.47914263656751\\
45	12.86055639217\\
};
\addlegendimage{only marks,color=red, line width=2.0pt, draw=none, mark=square, mark options={solid, red}, forget plot}
\addlegendentry{Simu, $N_s=2000$}
 \end{axis}

\begin{axis}[%
width=2.911in,
height=2.051in,
at={(0in,0in)},
scale only axis,
xmin=0,
xmax=1,
ymin=0,
ymax=1,
axis line style={draw=none},
ticks=none,
axis x line*=bottom,
axis y line*=left,
legend style={legend cell align=left, align=left, draw=white!15!black}
]
\end{axis}
\end{tikzpicture}%
\caption{\small{First moment of the conditional mean PAoI $M^1(\bar{A}_k)$ for different number of LEOSATs $N_s$ with a LEOSAT altitude $a_s =1500$ km.}}
\label{Fig:AoI_vs_Different_Theta_b_as_1500}
\vspace{0.2in}
\end{figure}
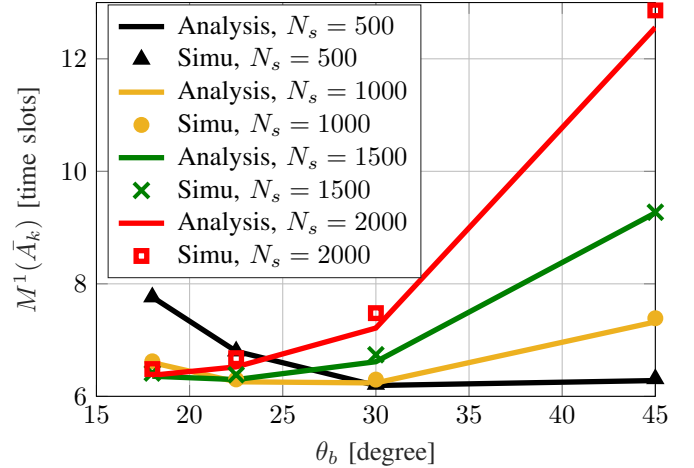

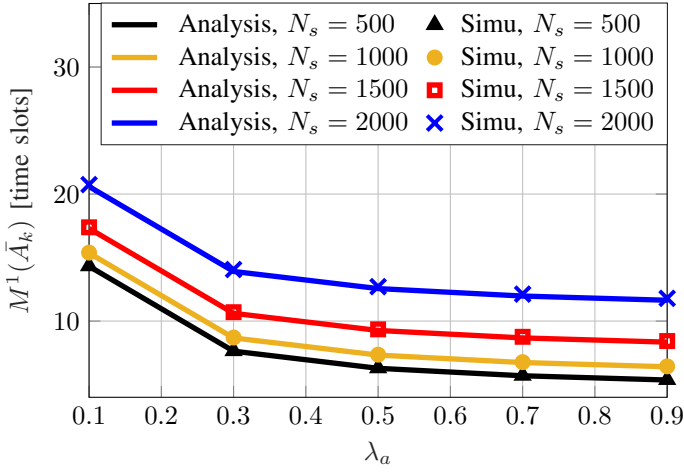
\begin{figure}[t]
\centering
%
%
\definecolor{mycolor1}{rgb}{0.92941,0.69412,0.12549}%
\begin{tikzpicture}

\begin{axis}[%
width=3.011in,
height=2.051in,
at={(0in,0in)},
scale only axis,
xmin=0.1,
xmax=0.9,
xlabel style={font=\color{white!15!black}},
xlabel={$\lambda_a$},
ymin=4,
ymax=35,
ylabel style={font=\color{white!15!black}},
ylabel={$M^1(\bar{A_k})$ [time slots]},
axis background/.style={fill=white},
xmajorgrids,
xminorgrids,
ymajorgrids,
yminorgrids,
legend style={at={(0.02,0.64)}, anchor=south west,legend cell align=left, align=left,legend columns=2,column sep=5pt, draw=white!15!black}
]
\addplot [color=black, line width=2.0pt]
  table[row sep=crcr]{%
0.1	14.3134884012289\\
0.3	7.63026300235801\\
0.5	6.28066661406484\\
0.7	5.69345751121995\\
0.9	5.36087124364098\\
};
\addlegendentry{Analysis, $N_s=500$}

\addplot [only marks,color=black, line width=2.0pt, draw=none, mark=triangle, mark options={solid, black}]
  table[row sep=crcr]{%
0.1	14.342154593663\\
0.3	7.65892919479203\\
0.5	6.30933280649886\\
0.7	5.72212370365398\\
0.9	5.389537436075\\
};
\addlegendentry{Simu, $N_s=500$}

\addplot [color=mycolor1, line width=2.0pt, forget plot]
  table[row sep=crcr]{%
0.1	15.3555958792041\\
0.3	8.67237048033321\\
0.5	7.32277409204004\\
0.7	6.73556498919515\\
0.9	6.40297872161618\\
};
\addlegendimage{color=mycolor1, line width=2.0pt, forget plot}
\addlegendentry{Analysis, $N_s=1000$}
\addplot [only marks,color=mycolor1, line width=2.0pt, draw=none, mark=*, mark options={solid, mycolor1}, forget plot]
  table[row sep=crcr]{%
0.1	15.387484890056\\
0.3	8.70425949118506\\
0.5	7.35466310289189\\
0.7	6.767454000047\\
0.9	6.43486773246803\\
};
\addlegendimage{only marks,color=mycolor1, line width=2.0pt, draw=none, mark=*, mark options={solid, mycolor1}, forget plot}
\addlegendentry{Simu, $N_s=1000$}
\addplot [color=red, line width=2.0pt, forget plot]
  table[row sep=crcr]{%
0.1	17.2941915053068\\
0.3	10.6109661064359\\
0.5	9.26136971814269\\
0.7	8.67416061529781\\
0.9	8.34157434771884\\
};
\addlegendimage{color=red, line width=2.0pt, forget plot}
\addlegendentry{Analysis, $N_s=1500$}
\addplot [only marks,color=red, line width=2.0pt, draw=none, mark=square, mark options={solid, red}, forget plot]
  table[row sep=crcr]{%
0.1	17.382201286942\\
0.3	10.6989758880711\\
0.5	9.34937949977795\\
0.7	8.76217039693306\\
0.9	8.42958412935409\\
};
\addlegendimage{only marks,color=red, line width=2.0pt, draw=none, mark=square, mark options={solid, red}, forget plot}
\addlegendentry{Simu, $N_s=1500$}
\addplot [color=blue, line width=2.0pt, forget plot]
  table[row sep=crcr]{%
0.1	20.5904468615536\\
0.3	13.9072214626827\\
0.5	12.5576250743895\\
0.7	11.9704159715446\\
0.9	11.6378297039657\\
};
\addlegendimage{color=blue, line width=2.0pt, forget plot}
\addlegendentry{Analysis, $N_s=2000$}
\addplot [only marks,color=blue, line width=1.5pt, draw=none, mark size=4.0pt, mark=x, mark options={solid, blue}, forget plot]
  table[row sep=crcr]{%
0.1	20.7419475593625\\
0.3	14.0587221604915\\
0.5	12.7091257721984\\
0.7	12.1219166693535\\
0.9	11.7893304017745\\
};
\addlegendimage{only marks,color=blue, line width=1.5pt, draw=none, mark size=4.0pt, mark=x, mark options={solid, blue}, forget plot}
\addlegendentry{Simu, $N_s=2000$}
\end{axis}

\begin{axis}[%
width=2.911in,
height=2.051in,
at={(0in,0in)},
scale only axis,
xmin=0,
xmax=1,
ymin=0,
ymax=1,
axis line style={draw=none},
ticks=none,
axis x line*=bottom,
axis y line*=left,
legend style={legend cell align=left, align=left, draw=white!15!black}
]
\end{axis}
\end{tikzpicture}%
\caption{\small{First moment of the conditional mean PAoI $M^1(\bar{A}_k)$ for different number of LEOSATs $N_s$ at a LEOSAT altitude $a_s =1500$ km.}}
\label{Fig:AoI_Different_diff_Ns_vs_Arrival_Rates_a_s_1500}
\vspace{0.2in}
\end{figure}

\begin{table}[h]\caption{Simulation Parameters}
\centering
\begin{center}\resizebox{\columnwidth}{!}{%
    \begin{tabular}{| c |c | c |}
    \hline
    \textbf{Notation}               & \textbf{Parameter}                                    & \textbf{Value} \\ \hline
    $r_e$                           & The Earth's radius                                    & 6371 km\\ \hline
    $\theta_b$                      & The beamwidth angle                                   & $\pi/4$\\ \hline
    $p_t^{IoT}$                     & The transmitted power of the IoT node                 & 23 dBm \\ \hline
    $\sigma_S^2$                    & The noise power at the LEOSAT                         & -174 dBm\\ \hline
    $f_c$                           & The frequency spectrum of the transmitted signal      & 2 GHz\\ \hline
    $B$                           & The channel bandwidth      & 10 kHz\\ \hline
    $\zeta$    & Channel gap factor & 0.8 \\ \hline
    $T_s$    & Time slot duration & 1 ms \\ \hline
    $c$                             & The speed of light                                    & $3 \times10^8$ m/s \\ \hline
    $\gamma_{th}$                   & The ${\rm{SINR}}$ threshold                           & $-20$ dB \\ \hline
    $\mathcal{SR}(\Omega,b_0,m)$    & Shadowed Rician Fading: light                         & $\mathcal{SR}(1.29,0.158,19.4)$ \\ \hline
    $\lambda_a$                     & Arrival rate of the updates                           & $0.5$ \\ \hline  
    \end{tabular}%
    }
\end{center}
\label{tab:Simulation_Parameters}
\vspace{-4mm}
\end{table}

We validate the analysis by conducting Monte Carlo simulations across $10^3$ different network realizations over $10^4$ time iterations, using the simulation parameter values specified in Table \ref{tab:Simulation_Parameters}. 

In Figs. \ref{Fig:P_Cov_Meta_Distribution_Ns_900} and \ref{Fig:P_Cov_Meta_Distribution_Ns_1300}, we plot the meta distribution of the conditional coverage probability at the typical LEOSAT at distinct LEOSAT altitudes $a_s$ (from 500 to 1500 km) for $N_s = 900$ and $N_s = 1300$, respectively. As shown in Fig. \ref{Fig:P_Cov_Meta_Distribution_Ns_900} and Fig. \ref{Fig:P_Cov_Meta_Distribution_Ns_1300}, increasing the LEOSAT altitude $a_s$ reduces the percentile of the update transmissions that exceed a certain reliability threshold $\tau$. This behavior can be justified as follows. As shown in \eqref{Eq:A_Cap}, the area of the spherical cap \(|{\mathcal A_{\rm Cap}}|\) is directly proportional to the altitude of the LEOSAT $a_s$. Accordingly, increasing $a_s$ guarantees involving a larger number of LEOSATs inside the spherical cap, leading to an increase in the number of interferes, which reduces the coverage probability at the typical LEOSAT, thereby minimizing the overall reliability of the network. Additionally, increasing the number of LEOSATs from $N_s = 900$ to $N_s = 1300$ leads to a decrease in the number of update transmissions across distinct network realizations that exceeds a certain reliability threshold $\tau$. This can be observed by comparing Fig. \ref{Fig:P_Cov_Meta_Distribution_Ns_900} with Fig. \ref{Fig:P_Cov_Meta_Distribution_Ns_1300}. As demonstrated in Figs. \ref{Fig:P_Cov_Meta_Distribution_Ns_900}, \ref{Fig:P_Cov_Meta_Distribution_Ns_1300}, \ref{Fig:PAoI_Distribution_Ns_900}, and \ref{Fig:PAoI_Distribution_Ns_1300}, there is a minor gap between the analytical results and the simulation results for the meta-distributions of both coverage probability and the spatial distribution of the conditional mean PAoI. This gap can be explained by two main assumptions made in this study: a Gamma distribution is used to model the shadowed Rician fading, while a Beta distribution is employed to calculate the coverage probability distribution. Furthermore, this gap is minimal when considering the lower altitudes of LEOSAT, and it is also small at higher LEOSAT altitudes. This emphasizes how well these assumptions align with the exact values.\\

Figs. \ref{Fig:PAoI_Distribution_Ns_900} and \ref{Fig:PAoI_Distribution_Ns_1300} show the spatial distribution of the conditional mean PAoI at the typical LEOSAT at distinct LEOSAT altitudes $a_s$ (from 500 to 1500 km) for the total number of LEOSATs $N_s = 900$ and $N_s = 1300$, respectively. As illustrated in both figures, averaging across different network realizations, the probability of achieving a certain mean PAoI decreases with increasing LEOSAT altitudes $a_s$. In addition to the increase in the number of potential interferers due to the higher LEOSAT altitudes $a_s$ as we mentioned, the propagation delay increases, which leads to a maximization in the mean PAoI of the transmitted update at the typical LEOSAT. On the other hand, increasing the total number of LEOSATs minimizes the probability of achieving a certain mean PAoI at the same LEOSAT altitude $a_s$, which is shown by comparing Fig. \ref{Fig:PAoI_Distribution_Ns_900} with Fig. \ref{Fig:PAoI_Distribution_Ns_1300}. This is because an increase in the total number of LEOSATs leads to an increase in the number of potential interferers, which subsequently increases the mean PAoI.

In Figs. \ref{Fig:Avg_PAoI_vs_h} and \ref{Fig:Avg_PAoI_vs_Ns}, we plot the first moment of the conditional mean PAoI for distinct numbers of LEOSATs $N_s$ (from 500 to 2000 LEOSATs) at different LEOSAT altitudes $a_s$ (from 500 to 1500 km). As illustrated, increasing the total number of LEOSATs at a higher LEOSAT altitude increases the spatial average of the conditional mean PAoI. This observation aligns with our earlier figures, which show that increased numbers of satellites and higher altitudes lead to an increase in the number of potential interferers, consequently reducing the coverage probability and, in turn, the mean of the PAoI.

Moreover, Fig. \ref{Fig:Avg_PAoI_vs_h} demonstrates an interesting trend for the number of  LEOSATs $N_s = 500$ and $N_s = 700$ (represented by the black and orange curves). It illustrates that the spatial average of the conditional mean PAoI can be minimized by considering a higher LEOSAT altitude. That means the vacation probability $P_{v}$ defined in \eqref{Eq:vacation_prob} has the dominant effect on the mean PAoI in these cases. This leads to an increase in the LEOSAT altitude, which guarantees the LEOSATs' availability to establish a connection to the intended node. Consequently, reducing the waiting time to begin serving the update and decreases the spatial average of the conditional mean PAoI.
Furthermore, for the case of $N_s = 900$ (depicted by the green curve), it is worth noting that there is an optimal LEOSAT altitude $a_s = 700$ km that gives the minimum spatial average of the conditional mean PAoI. Additionally, considering the LEOSAT altitude $a_s = 500$ km, it is beneficial to consider a higher number of LEOSATs to reduce the spatial average of the conditional mean PAoI, as shown by the black curve in Fig. \ref{Fig:Avg_PAoI_vs_Ns}. This behavior is justified as follows. In the lower altitudes, the size of the spherical cap is smaller because they are directly proportional. Therefore, increasing the number of LEOSATs increases the chance of including a larger number of LEOSATs within the spherical cap, which subsequently reduces the waiting time to begin service. 

In Figs. \ref{Fig:AoI_vs_Different_Theta_b} and \ref{Fig:AoI_vs_Different_Theta_b_as_1500}, we illustrate the impact of different LEOSAT beamwidth $\theta_b = [\pi/10, \pi/8, \pi/6, \pi/4]$ radian on the first moment of the conditional mean PAoI with a distinct number of LEOSATs $N_s = [500, 1000, 1500, 2000]$ at two different values of the LEOSAT altitude $a_s=500$ km and $a_s=1500$ km, respectively. As shown in Fig. \ref{Fig:AoI_vs_Different_Theta_b}, using a narrower beamwidth does not reduce the average of the conditional mean PAoI at a low LEOSAT altitude of $a_s=500$ km. On the other hand, at higher LEOSAT altitudes $a_s=1500$ km, applying a narrower beamwidth can effectively minimize the average of the conditional mean PAoI, as illustrated in Fig. \ref{Fig:AoI_vs_Different_Theta_b_as_1500}. The opposite trends in Fig. \ref{Fig:AoI_vs_Different_Theta_b} and Fig. \ref{Fig:AoI_vs_Different_Theta_b_as_1500} can be justified as follows. As shown in \eqref{Eq:A_Cap}, the area of the spherical cap \(|{\mathcal A_{\rm Cap}}|\) is directly proportional to the altitude of the LEOSAT $a_s$. Hence, for a lower LEOSAT altitude $a_s=500$ km, \(|{\mathcal A_{\rm Cap}}|\) involves a lower number of LEOSATs, which refers to a lower number of interferers. Therefore, increasing the LEOSAT footprint by considering a wider beamwidth increases the probability of establishing a connection between the typical LEOSAT and the intended source node. This results in a shorter waiting time to begin serving the current update and consequently minimizes the average of the conditional mean PAoI. On the other hand, for a LEOSAT altitude $a_s=1500$ km, the area of the spherical cap \(|{\mathcal A_{\rm Cap}}|\) expands to involve a larger number of LEOSATs that increases the number of potential interferers, which reduces the coverage probability at the typical LEOSAT. As a result, the service time for the current update increases, leading to a higher average of the conditional mean PAoI. Consequently, applying a narrower LEOSAT beamwidth at a higher LEOSAT altitude is beneficial to limit the number of interferers and minimize the average of the conditional mean PAoI. It is important to note that, as shown in Fig. \ref{Fig:AoI_vs_Different_Theta_b_as_1500}, for the scenario involving 500 LEOSATs at an altitude of 1500 km, the mean PAoI remains constant once the LEOSAT beamwidth reaches $\theta_b=30$ degrees. Therefore, setting the LEOSAT beamwidth to $\theta_b=45$ degrees is the ideal design point to maintain the antenna hardware cost without compromising data freshness. On the other hand, as the constellation density increases (with 1000, 1500, and 2000 LEOSATs), interference becomes a significant issue. In this context, wider LEOSAT beamwidths lead to longer service times because they capture more interfering signals. Consequently, denser constellations require narrower LEOSAT beamwidths to manage interference effectively and optimize data freshness. This results in distinct optimal LEOSAT beamwidths: $\theta_b=30$ degrees for 1000 LEOSATs, $\theta_b=22.5$ degrees for 1500 LEOSATs, and $\theta_b=18$ degrees for 2000 LEOSATs.

In Fig. \ref{Fig:AoI_Different_diff_Ns_vs_Arrival_Rates_a_s_1500}, we illustrate the impact of different arrival rates $\lambda_a = [0.1, 0.9]$ on the average of the conditional mean PAoI with a distinct number of LEOSATs $N_s = [500, 1000, 1500, 2000]$ at a typical LEOSAT altitude $a_s=1500$ km. As demonstrated in the figure, an increase in the arrival rate reduces the average of the conditional mean PAoI, as the source node continually transmits more recent updates.

\section{Conclusion} \label{Sec:Conclusion}
 
In this work, we derive the spatial distribution of the conditional mean PAoI for updates transmitted from a ground IoT node to its serving LEOSAT. Using tools from stochastic geometry, we begin our analysis by calculating the b-th moment of the conditional uplink coverage probability, which is essential for formulating the conditional coverage probability meta-distribution. From these expressions, we not only derive the spatial distribution of the conditional mean PAoI but also formulate the b-th moment expression for the conditional mean PAoI. We demonstrate how the varying number of LEOSATs and the altitude of the LEOSATs influence the first moment of the conditional mean PAoI, the conditional coverage probability meta-distribution, and the spatial distribution of the conditional mean PAoI. Moreover, to propose a deeper understanding of system performance, we evaluate distinct system setups by varying the total number of LEOSATs alongside their altitudes. Analyzing these diverse scenarios provides powerful insights into how the LEOSAT beamwidth impacts the conditional mean PAoI. Consequently, these findings enable system designers to precisely tune and identify the ideal beamwidth for LEOSATs, ensuring optimal performance that minimizes the conditional mean PAoI.

\appendices  
\section{Proof of Theorem~\ref{thm:bth_Moment_of_Conditional_Coverage_Probability}}\label{app:bth_Uplink_Pcov}

 \begin{align}\label{Eq:bth_Uplink_Pcov} 
	 & M^b(\gamma_{\rm th}|\psi_n) = \mathbb{E}_{\Psi_{n},|H|^2}[({{\rm{SINR}} \rm} >\gamma_{\rm th}| \Psi_{n})^b] \nonumber \\
    &= \mathbb{E}_{\Psi_{n},|H_o|^2}\left(\dfrac{ {\rho_t^{\rm IoT}} \left(\dfrac{\lambda_c}{4 \pi r_o}\right)^2 |H_o|^2}{ {\mathcal I_{\rm agg}^{\rm U}}+\sigma_S^2} >\gamma_{\rm th} | \Psi_{n}\right) \nonumber \\
    &\stackrel{(a)}{=} \mathbb{E}_{\Psi_{n}} \left[\mathbb{E}_{|H_o|^2}\left[ \left(|H_o|^2 > \dfrac{\gamma_{\rm th}({\mathcal I_{\rm agg}^{\rm U}}+\sigma_S^2) r_o^2}{{\rho_t^{\rm IoT}} \nu} | \Psi_{n} \right)^b \right] \right] \nonumber \\
    &\stackrel{(b)}{\approx} \mathbb{E}_{\Psi_{n}} \left[ \left(1- \left(1-\exp\left(\dfrac{-\mu_1 \gamma_{\rm th}({\mathcal I_{\rm agg}^{\rm U}}+\sigma_S^2) r_o^2  }{m_h{\rho_t^{\rm IoT}} \nu}\right)\right)^k \right)^b  \right] \nonumber \\
    &\stackrel{(c)}{=} \mathbb{E}_{\Psi_{n}} \biggl[ \biggl(\sum_{j=1}^k \binom{k}{j}(-1)^{j+1} \exp\biggl(\dfrac{-j\mu_1 \gamma_{\rm th}({\mathcal I_{\rm agg}^{\rm U}}+\sigma_S^2) r_o^2  }{m_h{\rho_t^{\rm IoT}} \nu}\biggr)\biggr)^b  \biggr] \nonumber \\
    &\stackrel{(d)}{=} \mathbb{E}_{\Psi_{n}} \biggl[ \sum_{\substack{b_1+b_2+\cdots +b_k=b \\ b_1,b_2,\cdots,b_k \geq 0}} \binom{b}{b_1,b_2,\cdots, b_k} \prod_{j=1}^k \biggl(\binom{k}{j} (-1)^{j+1} \times \nonumber \\
    &\exp\left(\dfrac{-j\mu_1 \gamma_{\rm th}{\mathcal I_{\rm agg}^{\rm U}} r_o^2}{m_h{\rho_t^{\rm IoT}} \nu}\right)\exp\biggl(\dfrac{-j\mu_1 \gamma_{\rm th}\sigma_S^2 r_o^2  }{m_h{\rho_t^{\rm IoT}} \nu}\biggr)\biggr) ^{b_j} \biggr] \nonumber \\
    &\stackrel{(e)}{=} \mathbb{E}_{r_o} \biggl[  \sum_{\substack{b_1+b_2+\cdots +b_k=b \\ b_1,b_2,\cdots,b_k \geq 0}} \binom{b}{b_1,b_2,\cdots, b_k} \times  \nonumber \\
    &\prod_{j=1}^k \biggl(\binom{k}{j} (-1)^{j+1} \biggr)^{b_j} 
     \exp\biggl(\dfrac{-j\mu_1 \gamma_{\rm th}\sigma_S^2 r_o^2 \sum_{j=1}^k j b_j}{m_h{\rho_t^{\rm IoT}} \nu}\biggr) \times \nonumber \\
    &\prod_{j=1}^k \bigg(\mathbb{E}_{{\mathcal I_{\rm agg}^{\rm U}}} \biggl[\exp\left(\dfrac{-j\mu_1 \gamma_{\rm th}{\mathcal I_{\rm agg}^{\rm U}} r_o^2}{m_h{\rho_t^{\rm IoT}} \nu}\right) \biggr]\biggr)^{b_j}\biggr] \nonumber \\
    &\stackrel{(f)}{=} \int_{a_s}^{ r_{\rm max}} \biggl[  \sum_{\substack{b_1+b_2+\cdots +b_k=b \\ b_1,b_2,\cdots,b_k \geq 0}} \binom{b}{b_1,b_2,\cdots, b_k} \times  \nonumber \\
    &\prod_{j=1}^k \biggl(\binom{k}{j} (-1)^{j+1} \biggr)^{b_j} 
    \exp\biggl(\dfrac{-j\mu_1 \gamma_{\rm th}\sigma_S^2 r_o^2 \sum_{j=1}^k j b_j}{m_h{\rho_t^{\rm IoT}} \nu}\biggr) \times \nonumber \\
    &\prod_{j=1}^k \biggl({\mathcal{L}_{\mathcal I_{\rm agg}^{\rm U}}} \biggl(\dfrac{-j\mu_1 \gamma_{\rm th}{\mathcal I_{\rm agg}^{\rm U}} r_o^2}{m_h{\rho_t^{\rm IoT}} \nu}\biggr)\biggr)^{b_j}\biggr] f_{R_o}(r_o) dr_o,
\end{align}
where (a) follows the independence of locations and channel power gains, where $\nu=(\lambda_c/4 \pi)^2$, (b) follows from approximating $|H_o|^2$ as a gamma distributed random variable \cite{Gamma_Bounds,Alzenad}  with shape parameter $k=\frac{m(2b_0+\Omega)^2}{4b_0^2m+4b_0m\Omega+\Omega^2}$ and scale parameter $m_h = \frac{4b_0^2m+4b_0m\Omega+\Omega^2}{m(2b_0+\Omega)^2}$, and $\mu_1$ is the upper bound parameter of the tight bounds for the Gamma function in \eqref{Eq:tight_bounds}, (c)  is the binomial theorem expansion of (b). In (d), we use the multinomial theorem with coefficients $b_1,b_2,\cdots,b_k \geq 0$, and $b_1+b_2+\cdots +b_k=b$ based on the shape parameter $k$, (e) follows the independence of the distances from the intended source node $r_o$ and the interferer nodes $r_i$ to the typical LEOSAT. Finally, (f) follows the definition of the Laplace transform as, ${\mathcal{L}_{\mathcal I_{\rm agg}^{\rm U}}} = E_{\mathcal I_{\rm agg}^{\rm U}}[e^{\rm -s{\mathcal I_{\rm agg}^{\rm U}}}]$, where $s=\dfrac{j\mu_1\gamma_{\rm th} r_o^2}{m_h{\rho_t^{\rm IoT}} \nu}$ which completes the proof. 

\begin{align}
    \label{Eq:tight_bounds} 
    (1-\exp({\dfrac{-\mu_1 x}{m_h}}))^k  \leq {F_{|H|^2}}(x) \leq (1-\exp({\dfrac{-\mu_2 x}{m_h}}))^k, \\
    \mu_1 =
        \begin{cases} 
            1 & k > 1 \\
            (k!)^{\dfrac{-1}{k}} & k < 1  \nonumber
        \end{cases}  \quad
    \mu_2 =
        \begin{cases} 
            (k!)^{\dfrac{-1}{k}}  & k > 1 \\
            1 & k < 1 
        \end{cases} 
\end{align}
    
\section{Proof of Lemma~\ref{lemma:Uplink_Iagg}}\label{app:Uplink_Iagg}
    \begin{align}\label{app:Uplink_Laplace_teanform_Iagg}
        &{\mathcal{L}_{\mathcal I_{\rm agg}^{\rm U}}(s)} = E_{\mathcal I_{\rm agg}^{\rm U}}\biggl[e^{-s{\mathcal I_{\rm agg}^{\rm U}}}\biggr] \nonumber \\
        &\stackrel{(a)}{=} E_{\Psi_n} \biggl[\prod_{i\in\{\Psi_n \cap \text{A(footprint) \rm}\}/o} E_{|H|^2}\biggl[\exp\biggl(-s{\rho_t^{\rm IoT}} 
        \left(\dfrac{\nu}{ r_i^2}\right) \times \nonumber \\
        &|H_i|^2\biggr)\biggr]\biggr] \nonumber \\
        &\stackrel{(b)}= E_{\Psi_n} \biggl[\left(1+\dfrac{s{\rho_t^{\rm IoT}}\nu m_h}{r_i^2}\right)^{-k} \mathbbm{1}(a_s, {\rm  r_{\rm max}}(\theta_c)) \nonumber \\
        &+ \mathbbm{1}({ \rm r_{\rm max}}(\theta_c),(r_s+r_e))\biggr]^{N_s-1} \nonumber \\
        &= \biggl[\int_{a_s}^{\rm {r_{\rm max}}} \left(1+\left(\dfrac{s{\rho_t^{\rm IoT}}\nu m_h}{r_i^2}\right)\right)^{-k} f_{R_i}(r_i) dr_i \nonumber \\
        &+ \int_{\rm {r_{\rm max}}}^{(r_s+r_e)} f_{R_i}(r_i) dr_i\biggr]^{N_s-1},
    \end{align}
where (a) applies the exponential property, which allows us to replace the exponential of a sum with the product of exponential and follows the independence of locations and channel power gains, then using the exponential property, (b) considers the Moment Generating Function (MGF) of Gamma distribution with shape $k$ and scale $m_h$ which is defined as $M_{|H|^2}(t) = E[e^{t|H|^2}] = (1-m_h t)^{-k}$, and then splits the interference area into two regions: if the active IoT node lies within the bounded area from $a_s$ to ${ r_{\rm max}}$, it is considered a potential interferer; otherwise, it is not an interferer node. Please refer to \cite{Badiaa_LEO1} for a more detailed version of this derivation.

\section{Proof of Theorem~\ref{thm:bth_Moment_of_Conditional_Mean_PAoI}}\label{app:bth_Moment_of_Conditional_Mean_PAoI}

 \begin{align}\label{Eq:bth_Uplink_Conditional_Mean_PAoI} 
	 &M^b(\bar{A}_k) = \mathbb{E}[((T_{\rm k-1} +T_k + T^I_k + T^{\rm U}_P) |\Psi_n)^b] \nonumber \\
     &\stackrel{(a)}{=} \mathbb{E}_{\Psi_{n}} \left[ \sum_{x=0}^b \binom{b}{x} W^x Z^{b-x}\right] \nonumber \\
     &=  \sum_{x=0}^b \binom{b}{x} \left( \frac{2}{(1-P_{v})} + \frac{1}{\lambda_a}\right)^{b-x} \mathbb{E}_{\Psi_{n}} \biggl[ \biggl(\dfrac{2}{P^{\rm U}_{\rm cov}(\gamma_{\rm th}|\Psi_n)}   \nonumber \\
     &+ \dfrac{1}{c}\int_{a_s}^{ {\rm r_{\rm max}}} r_o f_{R_o}(r_o) dr_o\biggr)^x\biggr] \nonumber \\
    &\stackrel{(b)}{=} \sum_{x=0}^b \binom{b}{x} \left( \frac{2}{(1-P_{v})} + \frac{1}{\lambda_a}\right)^{b-x} \mathbb{E}_{\Psi_{n}} \biggl[ \sum_{y=0}^x \binom{x}{y} \times \nonumber \\
    & \biggl(\dfrac{2}{P^{\rm U}_{\rm cov}(\gamma_{\rm th}|\Psi_n)}\biggr)^x \biggl(\dfrac{1}{c}\int_{a_s}^{ {\rm r_{\rm max}}} r_o f_{R_o}(r_o) dr_o  \biggr)^{x-y}\biggr] \nonumber \\
    &\stackrel{(c)}{=} \sum_{x=0}^b \binom{b}{x} \left( \frac{2}{(1-P_{v})} + \frac{1}{\lambda_a}\right)^{b-x} \biggl[ \sum_{y=0}^x \binom{x}{y}  \times \nonumber \\
    &\biggl(\dfrac{2}{M^b(P^{\rm U}_{\rm cov}(\gamma_{\rm th}| \Psi_{n}))}\biggr)^x  \biggl(\dfrac{1}{c}\int_{a_s}^{ {\rm r_{\rm max}}} r_o f_{R_o}(r_o) dr_o \biggr)^{x-y}\biggr]
\end{align}
where (a) follows the binomial theorem expansion, where $W = \dfrac{2}{P^{\rm U}_{\rm cov}(\gamma_{\rm th})} + \dfrac{1}{c}\int_{a_s}^{ {\rm r_{\rm max}}} r_o f_{R_o}(r_o) dr_o$ and $Z = \dfrac{2}{(1-P_{v})} + \dfrac{1}{\lambda_a}$, (b) applies the linearity of expectation and the binomial theorem expansion, (c) is the expectation for a given network setup across different network realizations and $M^b(P^{\rm U}_{\rm cov}(\gamma_{\rm th}| \Psi_{n}))$ is given in \eqref{Eq:Uplink_Coverage_Probability_bth_Moments_Final}.

\section{Proof of Theorem~\ref{thm:Spatial_Distribution_of_Conditional_Mean_PAoI}}\label{app:Spatial_Distribution_of_Conditional_Mean_PAoI}

 \begin{align}\label{Eq:bth_Uplink_Conditional_Mean_PAoI} 
	 &F_{\bar{A}_k}(\gamma_{\rm th},t) = \mathbb{P}(A_k(\gamma_{\rm th}|\Psi_n) < t) \nonumber \\
     &= \mathbb{P}\biggl(\biggl(\dfrac{2}{P^{\rm U}_{\rm cov}(\gamma_{\rm th}|\Psi_n)} + \dfrac{2}{(1-P_{v})} + \dfrac{1}{\lambda_a}\nonumber \\
      &+ \dfrac{1}{c}\int_{a_s}^{ {\rm r_{\rm max}}} r_o f_{R_o}(r_o) dr_o \biggr)< t \biggr) \nonumber \\
      &= \mathbb{P}\biggl(P^{\rm U}_{\rm cov}(\gamma_{\rm th}|\Psi_n) > 2 \biggl( t - \dfrac{2}{(1-P_{v})} - \dfrac{1}{\lambda_a}\nonumber \\
      &- \dfrac{1}{c}\int_{a_s}^{ {\rm r_{\rm max}}} r_o f_{R_o}(r_o) dr_o \biggr)^{-1}\biggr) \nonumber \\
     &\stackrel{(a)}{=} 1 - \mathcal{I}_\omega \left(\alpha,\beta\right),
\end{align}
where (a) follows the Beta distribution approximation for the coverage probability distribution, $\mathcal{I}_\omega (\alpha,\beta)$ is the regularized incomplete beta function, $\alpha$ and $\beta$ are defined in \eqref{Eq:Beta}, and $\omega = 2\left(t-\dfrac{2}{(1-P_{v})} + \dfrac{1}{\lambda_a}+ \dfrac{1}{c}\int_{a_s}^{ {\rm r_{\rm max}}} r_o f_{R_o}(r_o) dr_o\right)^{-1}$.
\bibliographystyle{IEEEtran}
\bibliography{ref}

@article{AoI1,
  title={{Real-time status: How often should one update?}},
  author={Sanjit Krishnan Kaul and Roy D. Yates and Marco Gruteser},
  journal={  Proc. IEEE INFOCOM},
  year={2012},
  pages={2731-2735}
}

@ARTICLE{AoI8,
  author={Costa, Maice and Codreanu, Marian and Ephremides, Anthony},
  journal={IEEE Transactions on Information Theory}, 
  title={{On the Age of Information in Status Update Systems With Packet Management}}, 
  year={2016},
  volume={62},
  number={4},
  pages={1897-1910},
  doi={10.1109/TIT.2016.2533395}}

@INPROCEEDINGS{Howard,
  author={Lu, Yanwu and Yang, Howard H. and Pappas, Nikolaos and Geraci, Giovanni and Ma, Chuan and Quek, Tony Q. S.},
  booktitle={ Proc. IEEE GC Wkshps}, 
  title={{Analysis of Age of Information in Non-terrestrial Networks}}, 
  year={2023},
  volume={},
  number={},
  pages={239-244},
  doi={10.1109/GCWkshps58843.2023.10465028}}

@INPROCEEDINGS{AoI18,
  author={Guo, Zheng and Chen, Qian and Meng, Weixiao},
  booktitle={  Proc. IEEE GLOBECOM }, 
  title={Service-Oriented {AoI} Modeling and Analysis for Non-Terrestrial Networks}, 
  year={2024},
  volume={},
  number={},
  pages={4316-4321},
  doi={10.1109/GLOBECOM52923.2024.10901619}}

@ARTICLE{Anna1,
  author={Talgat, Anna and Kishk, Mustafa A. and Alouini, Mohamed-Slim},
  journal={IEEE Transactions on Aerospace and Electronic Systems}, 
  title={{Stochastic Geometry-Based Uplink Performance Analysis of {IoT} over {LEO} Satellite Communication}}, 
  year={2024},
  volume={60},
  number={4},
  pages={4198-4213},
  doi={10.1109/TAES.2024.3374718}}

@article{Akram1,
  title={{An Analytic Approach for Modeling The Coverage Performance of Dense Satellite Networks}},
  author={Al-Hourani, Akram},
  journal={IEEE Wireless Communications Letters},
  volume={10},
  number={4},
  pages={897--901},
  year={2021},
  publisher={IEEE}
}

@ARTICLE{Niloofar2,
  author={Okati, Niloofar and Riihonen, Taneli},
  journal={IEEE Transactions on Communications}, 
  title={{Nonhomogeneous Stochastic Geometry Analysis of Massive {LEO} Communication Constellations}}, 
  year={2022},
  volume={70},
  number={3},
  pages={1848-1860},
  doi={10.1109/TCOMM.2022.3143131}}

@ARTICLE{Namyoon1,
  author={Park, Jeonghun and Choi, Jinseok and Lee, Namyoon},
  journal={IEEE Transactions on Wireless Communications}, 
  title={{A Tractable Approach to Coverage Analysis in Downlink Satellite Networks}}, 
  year={2023},
  volume={22},
  number={2},
  pages={793-807},
  doi={10.1109/TWC.2022.3198103}}

@ARTICLE{Choi1,
  author={Choi, Chang-Sik and Baccelli, François},
  journal={IEEE Transactions on Vehicular Technology}, 
  title={{Cox Point Processes for Multi Altitude {LEO} Satellite Networks}}, 
  year={2024},
  volume={73},
  number={10},
  pages={15916-15921},
  doi={10.1109/TVT.2024.3414833}}

@ARTICLE{Choi2,
  author={Choi, Chang-Sik and Baccelli, François},
  journal={IEEE Transactions on Communications}, 
  title={{A Novel Analytical Model for {LEO} and MEO Satellite Networks Based on Cox Point Processes}}, 
  year={2024},
  volume={},
  number={},
  pages={1-1},
  doi={10.1109/TCOMM.2024.3471978}}

@ARTICLE{Akram2,
  author={Homssi, Bassel Al and Al-Hourani, Akram},
  journal={IEEE Communications Letters}, 
  title={{Modeling Uplink Coverage Performance in Hybrid Satellite-Terrestrial Networks}}, 
  year={2021},
  volume={25},
  number={10},
  pages={3239-3243},
  doi={10.1109/LCOMM.2021.3103942}}

@ARTICLE{Badiaa_LEO1,
  author={Gabr, Badiaa and Kishk, Mustafa A.},
  journal={IEEE Transactions on Wireless Communications}, 
  title={{Performance Analysis of Average Peak Age of Information in LEO Satellite-Enabled IoT Networks}}, 
  year={2026},
  volume={25},
  number={},
  pages={5186-5199},
  doi={10.1109/TWC.2025.3616774}}

@ARTICLE{Jeffrey,
  author={Andrews, Jeffrey G. and Baccelli, Francois and Ganti, Radha Krishna},
  journal={IEEE Transactions on Communications}, 
  title={{A Tractable Approach to Coverage and Rate in Cellular Networks}}, 
  year={2011},
  volume={59},
  number={11},
  pages={3122-3134},
  doi={10.1109/TCOMM.2011.100411.100541}}

@book{Haenggi2012,
  title={{Stochastic Geometry for Wireless Networks}},
  author={Martin Haenggi},
  year={2012},
  publisher = {Cambridge University Press}
}

@ARTICLE{Alouini_1,
  author={Abdi, A. and Lau, W.C. and Alouini, M.-S. and Kaveh, M.},
  journal={IEEE Transactions on Wireless Communications}, 
  title={{A New Simple Model for Land Mobile Satellite Channels: First- And Second-Order Statistics}}, 
  year={2003},
  volume={2},
  number={3},
  pages={519-528},
  doi={10.1109/TWC.2003.811182}}

@book{Table_of_integrals,
  title={Table of integrals, series, and products},
  author={Gradshteyn, Izrail Solomonovich and Ryzhik, Iosif Moiseevich},
  year={2014},
  publisher={New York, NY, USA: Academic}
}

@ARTICLE{Alzenad,
  author={Alzenad, Mohamed and Yanikomeroglu, Halim},
  journal={IEEE Transactions on Wireless Communications}, 
  title={{Coverage and Rate Analysis for Vertical Heterogeneous Networks (VHetNets)}}, 
  year={2019},
  volume={18},
  number={12},
  pages={5643-5657},
  doi={10.1109/TWC.2019.2938168}}

@article{Gamma_Bounds,
author = {Alzer, Horst},
title = {{On Some Inequalities for The Incomplete Gamma Function}},
year = {1997},
issue_date = {April 1997},
publisher = {American Mathematical Society},
address = {USA},
volume = {66},
number = {218},
issn = {0025-5718},
doi = {10.1090/S0025-5718-97-00814-4},
journal = {Math. Comput.},
month = {apr},
pages = {771–778},
numpages = {8}
}

@INPROCEEDINGS{AoI_LEO_01,
  author={Lin, Chih-Yu and Liao, Wanjiun},
  booktitle={ Proc. IEEE VTC2024-Spring}, 
  title={Energy-aware Age of Information {(AoI)} Minimization for Internet of Things in {NOMA}-based {LEO} Satellite Networks}, 
  year={},
  volume={},
  number={},
  pages={1-5},
  doi={10.1109/VTC2024-Spring62846.2024.10683037}}

@ARTICLE{Anna4,
  author={Talgat, Anna and Kishk, Mustafa A. and Alouini, Mohamed-Slim},
  journal={IEEE Communications Letters}, 
  title={Stochastic Geometry-Based Analysis of {LEO} Satellite Communication Systems}, 
  year={2021},
  volume={25},
  number={8},
  pages={2458-2462},
  doi={10.1109/LCOMM.2020.3029808}}

@ARTICLE{AoI_23,
  author={Zhang, Di and Peng, Yuanyuan and Ke, Ying and Song, Jinpeng and Wang, Shuai and Pan, Gaofeng and An, Jianping},
  journal={IEEE Transactions on Vehicular Technology}, 
  title={Age of Information Variation of {LEO} Satellite-Terrestrial Uplink Transmissions}, 
  year={2025},
  volume={74},
  number={6},
  pages={9645-9655},
  doi={10.1109/TVT.2025.3540572}}

@ARTICLE{Meta_1,
  author={Angervuori, Ilari and Haenggi, Martin and Wichman, Risto},
  journal={IEEE Transactions on Communications}, 
  title={Meta Distribution of the {SIR} in a Narrow-Beam {LEO} Uplink}, 
  year={2025},
  volume={73},
  number={9},
  pages={8260-8273},
  doi={10.1109/TCOMM.2025.3547734}}

@ARTICLE{Satellite4,
  author={Zhang, Jin and Xu, Lexi and Wei, Haichao and Deng, Na and Zhao, Nan and Nallanathan, Arumugam},
  journal={IEEE Transactions on Wireless Communications}, 
  title={Local Delay in {LEO} Satellite Mega-Constellations}, 
  year={2025},
  volume={},
  number={},
  pages={1-1},
  doi={10.1109/TWC.2025.3616286}}

@ARTICLE{Meta_2,
  author={Sun, Yanshi and Ding, Zhiguo},
  journal={IEEE Networking Letters}, 
  title={A Fine Grained Stochastic Geometry-Based Analysis on {LEO} Satellite Communication Systems}, 
  year={2023},
  volume={5},
  number={4},
  pages={237-240},
  doi={10.1109/LNET.2023.3269818}}

@INPROCEEDINGS{PPP,
  author={Sellathurai, Mathini and Vuppala, Satyanarayana and Ratnarajah, Tharmalingam},
  booktitle={ Proc. IEEE Asilomar}, 
  title={User selection for multi-beam satellite channels: A stochastic geometry perspective}, 
  year={2016},
  volume={},
  number={},
  pages={487-491},
  doi={10.1109/ACSSC.2016.7869087}}

@ARTICLE{Ruibo3,
  author={Wang, Ruibo and Kishk, Mustafa A. and Alouini, Mohamed-Slim},
  journal={IEEE Transactions on Aerospace and Electronic Systems}, 
  title={Stochastic Geometry-Based Low Latency Routing in Massive {LEO} Satellite Networks}, 
  year={2022},
  volume={58},
  number={5},
  pages={3881-3894},
  doi={10.1109/TAES.2022.3199682}}

@ARTICLE{Satellite5,
  author={Dong, Wen-Yu and Yang, Shaoshi and Zhang, Ping and Chen, Sheng},
  journal={IEEE Internet of Things Journal}, 
  title={Modeling and Performance Analysis of {IoT}-Over-{LEO} Satellite Systems Under Realistic Operational Constraints: A Stochastic Geometry Approach}, 
  year={2025},
  volume={12},
  number={15},
  pages={30576-30593},
  doi={10.1109/JIOT.2025.3570843}}

@ARTICLE{Satellite6,
  author={Li, Yue and Xu, Yunlai and Zhang, Qinyu and Yang, Zhihua},
  journal={IEEE Internet of Things Journal}, 
  title={Age-Optimized Multihop {Information} Update Mechanism on the {LEO} Satellite Constellation via Continuous Time-Varying Graphs}, 
  year={2023},
  volume={10},
  number={8},
  pages={7189-7203},
  doi={10.1109/JIOT.2022.3229028}}

@ARTICLE{Akram3,
  author={Al Homssi, Bassel and Al-Hourani, Akram},
  journal={IEEE Wireless Communications Letters}, 
  title={Optimal Beamwidth and Altitude for Maximal Uplink Coverage in Satellite Networks}, 
  year={2022},
  volume={11},
  number={4},
  pages={771-775},
  doi={10.1109/LWC.2022.3143498}}

@ARTICLE{Tony_Quek,
  author={Wei, Qing and Shi, Jia and Li, Zan and Zhao, Zhongling and Quek, Tony Q.S.},
  journal={IEEE Transactions on Vehicular Technology}, 
  title={{AoI} and Energy Minimization for {LEO} Satellite-Terrestrial Networks: A Constrained Multi-Objective Optimization Approach}, 
  year={2025},
  volume={74},
  number={8},
  pages={12436-12448},
  doi={10.1109/TVT.2025.3553499}}

@ARTICLE{Cox1,
  author={Choi, Chang-Sik},
  journal={IEEE Transactions on Wireless Communications}, 
  title={Modeling and Analysis of Downlink Communications in a Heterogeneous {LEO} Satellite Network}, 
  year={2024},
  volume={23},
  number={8},
  pages={8588-8602},
  doi={10.1109/TWC.2024.3351876}}

@ARTICLE{Cox2,
  author={Choi, Chang-Sik},
  journal={IEEE Journal on Selected Areas in Communications}, 
  title={Analysis of a Delay-Tolerant Data Harvest Architecture Leveraging Low Earth Orbit Satellite Networks}, 
  year={2024},
  volume={42},
  number={5},
  pages={1329-1343},
  doi={10.1109/JSAC.2024.3365871}}

@INPROCEEDINGS{Badiaa_LEO2,
  author={Gabr, Badiaa and Kishk, Mustafa A.},
  booktitle={Proc. IEEE VTC2025-Spring)}, 
  title={{Statistical Analysis for Average Peak Age of Information in {LEO} Satellite-Enabled IoT}}, 
  year={},
  volume={},
  number={},
  pages={1-6},
  doi={10.1109/VTC2025-Spring65109.2025.11174394}}

@INPROCEEDINGS{Meta_3,
  author={Naganjani, Baly and Dwivedi, Ayush Kumar and Chaudhari, Sachin and Riihonen, Taneli},
  booktitle={Proc. IEEE  Globecom Workshops (GC Wkshps)}, 
  title={Analyzing {6G} Satellite {IoT} Architecture Using Stochastic Geometry: A Meta Distribution Approach}, 
  year={2024},
  volume={},
  number={},
  pages={1-6},
  doi={10.1109/GCWkshp64532.2024.11101096}}

@ARTICLE{ElSawy1,
  author={ElSawy, Hesham and Alouini, Mohamed-Slim},
  journal={IEEE Communications Letters}, 
  title={{On the Meta Distribution of Coverage Probability in Uplink Cellular Networks}}, 
  year={2017},
  volume={21},
  number={7},
  pages={1625-1628},
  doi={10.1109/LCOMM.2017.2692205}}

@ARTICLE{Martin1,
  author={Haenggi, Martin},
  journal={IEEE Transactions on Wireless Communications}, 
  title={{The Meta Distribution of the SIR in Poisson Bipolar and Cellular Networks}}, 
  year={2016},
  volume={15},
  number={4},
  pages={2577-2589},
  doi={10.1109/TWC.2015.2504983}}

@ARTICLE{Meta_4,
  author={Lin, Hongxin and Zhang, Cheng and Huang, Yongming and Zhao, Rui and Yang, Luxi},
  journal={IEEE Wireless Communications Letters}, 
  title={{Fine-Grained Analysis on Downlink LEO Satellite-Terrestrial mmWave Relay Networks}}, 
  year={2021},
  volume={10},
  number={9},
  pages={1871-1875},
  doi={10.1109/LWC.2021.3084626}}

@misc{semantics,
      title={{Semantics-Aware Unified Terrestrial Non-Terrestrial 6G Networks}}, 
      author={Erfan Delfani and Agapi Mesodiakaki and Leandros Tassiulas and Nikolaos Pappas},
      year={2026},
      eprint={2505.01796},
      archivePrefix={arXiv},
      primaryClass={cs.NI},
      url={https://arxiv.org/abs/2505.01796}, 
}

@ARTICLE{Akram4,
  author={Al Homssi, Bassel and Al-Hourani, Akram and Wang, Ke and Conder, Phillip and Kandeepan, Sithamparanathan and Choi, Jinho and Allen, Ben and Moores, Ben},
  journal={IEEE Communications Magazine}, 
  title={{Next Generation Mega Satellite Networks for Access Equality: Opportunities, Challenges, and Performance}}, 
  year={2022},
  volume={60},
  number={4},
  pages={18-24},
  doi={10.1109/MCOM.001.2100802}}

@ARTICLE{El_Hajjar_1,
  author={Li, Qingchao and El-Hajjar, Mohammed and Cao, Kaijun and Xu, Chao and Haas, Harald and Hanzo, Lajos},
  journal={IEEE Transactions on Wireless Communications}, 
  title={{Holographic Metasurface-Based Beamforming for Multi-Altitude LEO Satellite Networks}}, 
  year={2025},
  volume={24},
  number={4},
  pages={3103-3116},
  doi={10.1109/TWC.2025.3527962}}

@ARTICLE{El_Hajjar_2,
  author={Li, Qingchao and El-Hajjar, Mohammed and Xu, Chao and An, Jiancheng and Yuen, Chau and Hanzo, Lajos},
  journal={IEEE Transactions on Communications}, 
  title={{Stacked Intelligent Metasurface-Based Transceiver Design for Near-Field Wideband Systems}}, 
  year={2025},
  volume={73},
  number={9},
  pages={8125-8139},
  doi={10.1109/TCOMM.2025.3544929}}

@ARTICLE{El_Hajjar_3,
  author={Li, Qingchao and El-Hajjar, Mohammed and Sun, Yanshi and Hanzo, Lajos},
  journal={IEEE Transactions on Wireless Communications}, 
  title={{Performance Analysis of Reconfigurable Holographic Surfaces in the Near-Field Scenario of Cell-Free Networks Under Hardware Impairments}}, 
  year={2024},
  volume={23},
  number={9},
  pages={11972-11984},
  doi={10.1109/TWC.2024.3386850}}

@ARTICLE{AoI23,
  author={Yang, Howard H. and Xu, Chao and Wang, Xijun and Feng, Daquan and Quek, Tony Q. S.},
  journal={IEEE Transactions on Wireless Communications}, 
  title={{Understanding Age of Information in Large-Scale Wireless Networks}}, 
  year={2021},
  volume={20},
  number={5},
  pages={3196-3210},
  doi={10.1109/TWC.2020.3048008}}

@ARTICLE{Satellite7,
  author={Jung, Dong-Hyun and Ryu, Joon-Gyu and Byun, Woo-Jin and Choi, Junil},
  journal={IEEE Transactions on Communications}, 
  title={{Performance Analysis of Satellite Communication System Under the Shadowed-Rician Fading: A Stochastic Geometry Approach}}, 
  year={2022},
  volume={70},
  number={4},
  pages={2707-2721},
  doi={10.1109/TCOMM.2022.3142290}}

\end{document}